\documentclass[runningheads]{llncs}
\usepackage{amsmath}
\usepackage{amssymb}
\usepackage{amsfonts}
\usepackage{ifthen}
\usepackage{color,soul}
\usepackage{txfonts}
\usepackage[linesnumbered]{algorithm2e}
\usepackage{fullpage}
\usepackage{url}
\usepackage{tikz}
\usetikzlibrary{arrows,backgrounds,decorations,decorations.pathmorphing,
	positioning,fit,automata,shapes,snakes,patterns,plotmarks,calc,trees}
\usepackage{graphicx}
\usepackage{latexsym}	
\usepackage{xspace}
\usepackage{alltt}
\usepackage[obeyFinal]{todonotes}
\usepackage{wrapfig}
\usepackage{subcaption}
\usepackage[per-mode=symbol]{siunitx}
\usepackage{paralist}
\usepackage{color,colortbl}
\usepackage{pgfgantt}

\usepackage{epsfig,color}
\usepackage{algpseudocode}
\usepackage{mdwlist}
\usepackage{setspace} 
\usepackage{hyperref}
\usepackage[depth=1]{bookmark}
\usepackage{listings}
\usepackage{caption}
\usepackage{xspace}
\usepackage{tikz}
\usepackage{stmaryrd}
\usepackage{multirow}
\usepackage{multicol}

\usepackage{ltl}
\usepackage{habbas-macros}

\newcommand{\festl}{BB-STL}
\newcommand{\pastBSTL}{STL$^{\text{past}}_{\text{bdd}}$}
\newcommand{\Nlo}{\Nc} 
\newcommand{\Nneg}{\Nlo_-}
\newcommand{\Nmin}{\sqcap}
\newcommand{\Nmax}{\sqcup}
\newcommand{\Nonce}{\sqcup} 
\newcommand{\Nhist}{\sqcap}

\newcommand{\Fourier}[1]{\Fc\{#1\}}
\newcommand{\Omgn}{\vec{\Omega}}
\newcommand{\Taun}{\vec{\tau}}
\newcommand{\FilterSp}{\Dc_{1}'}  
\renewcommand{\sttraj}{x}
\renewcommand{\sstraj}{x}
\renewcommand{\outtraj}{y}
\renewcommand{\ztraj}{z}
\renewcommand{\TDom}{E}
\renewcommand{\SigSpace}{\Dc(\TDom)}
\renewcommand{\stSet}{\Xc}

\begin{document}
\title{Logical Signal Processing: a Fourier Analysis of Temporal Logic}
%
%
\author{Niraj Basnet\and
	Houssam Abbas}
\authorrunning{N. Basnet et al.}
%
\institute{Oregon State University, Corvallis OR 97330, USA
\email{\{basnetn,abbasho\}@oregonstate.edu}
}
\maketitle              
\begin{abstract}
What is the frequency content of temporal logic formulas? That is, when we monitor a signal against a formula, which frequency bands of the signal are relevant to the logic and should be preserved, and which can be safely discarded?
This question is relevant whenever signals are filtered or compressed before being monitored, which is almost always the case for analog signals.
To answer this question, we focus on monitors that measure the robustness of a signal relative to a specification in Signal Temporal Logic. 
We prove that robustness monitors can be modeled using Volterra series.
We then study the Fourier transforms of these Volterra representations, and provide a method to derive the Fourier transforms of entire formulas.
We also make explicit the measurement process in temporal logic and re-define it on the basis of distributions to make it compatible with measurements in signal processing.
Experiments illustrate these results.
Beyond compression, this work enables the integration of temporal logic monitoring into common signal processing tool chains as just another signal processing operation, and enables a common formalism to study both logical and non-logical operations in the frequency domain, which we refer to as Logical Signal Processing.
\end{abstract}

\section{Introduction: The Place of Runtime Verification in the Signal Processing Chain}
\label{sec:intro}

Runtime monitors in Cyber-Physical Systems (CPS) process analog signals: that is, continuous-time, continuous-valued signals generated by the physics of the system, rather than digital signals generated by computations on values stored in memory.
These analog signals are never pristine: to begin with, they are measured, and so incur some measurement distortion; 
they are noisy, and therefore are usually filtered to reduce noise;
and if they get transmitted, they are compressed and de-compressed, which introduces further distortions. 
All of these operations are very common in signal processing toolchains - indeed, the act of measurement is inevitable. 
And all of these operations affect, a priori, the verdict of the runtime monitor. 
Yet we have little theory to \textit{systematically} account for these effects. 
For instance, Fourier analysis is a standard powerful tool in signal processing, which is used to synthesize optimal filters meeting certain output requirements.
If we want to synthesize a filter subject to monitorability requirements, how would we go about it?
Conversely, how can we synthesize a monitor that accounts for the modifications introduced by a filter earlier in the processing chain?
Today we have no way of answering these questions systematically, because we lack an account of the frequency content of temporal logic. That is, we lack a Fourier representation of the operation of a temporal logic monitor.
Lacking such a uniform formalism, we remain in the awkward position of having to study the impact of frequency-domain operations in the time domain, using ad hoc assumptions like the availability of Lipschitz constants, and combining them somehow with the time-domain representation of monitors. 

This paper presents a way of analyzing the frequency content of temporal logic, or equivalently, of analyzing temporal logic monitors in the Fourier domain. 
Because monitors are nonlinear operators, we resort to Volterra series, which generalize the convolution representation of time-invariant linear operators to time-invariant nonlinear operators, and have generalized Fourier representations.
To apply the machinery of Volterra series, we work exclusively with 
robustness monitors, which output analog robustness signals.

As soon as we start thinking in the Fourier domain, we also realize that the basic measurement model in temporal logic is broken from a physical, and therefore, cyber-physical, perspective: in almost all temporal logics, it is assumed that one can make an \textit{instantaneous} measurement.
I.e., that it is possible to measure $x(t)$ exactly at $t$, and this is used to determine the truth value of an atomic proposition, e.g. `$x(t)\geq 0'$.
However, it is well-known that instantaneous measurements of analog signals are impossible!
Any measurement device has a finite resolution, so at best we can measure some aggregate of infinitely many signal values.
For instance, in a camera the value recorded by a pixel equals the \textit{average} illumination incident on that pixel, not the illumination in a specific point in space.
This matters to us because an instantaneous measurement requires infinite bandwidth, thus rendering the entire frequency analysis trivial or useless.
Instantaneous measurements also produce mathematical complications when working with sup norms and Lebesgue integrals which ignore sets of measure zero.
We therefore re-define atomic propositions (and the measurement model) on the basis of the theory of distributions, and demonstrate that the resulting robust semantics are still sound and still yield robustness tubes that can be used in falsification.

\begin{figure}[t]
	\centering
	\includegraphics[width=0.8\linewidth]{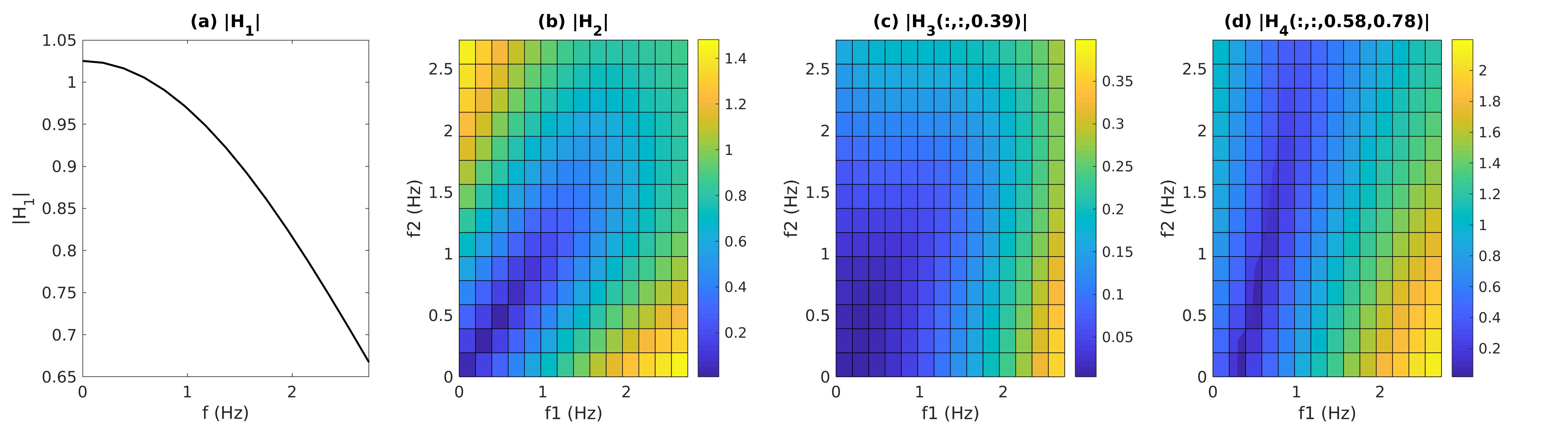}
	\caption{Frequency response of the temporal logic formula Historically $p$. $H_1(\omega)$ captures the linear, or first-order, transfer function, while $H_2(\omega_1,\omega_2)$, $H_3(\omega_1,\omega_2,\omega_3)$, etc, are higher-order generalized frequency responses that capture the non-linear effects. By studying these response functions we can determine which frequencies of the monitored signal affect the monitor output. (Color in digital copy).}
	\label{fig:example}
\end{figure}
\begin{figure}[t]
	\centering
	\includegraphics[width=0.8\linewidth]{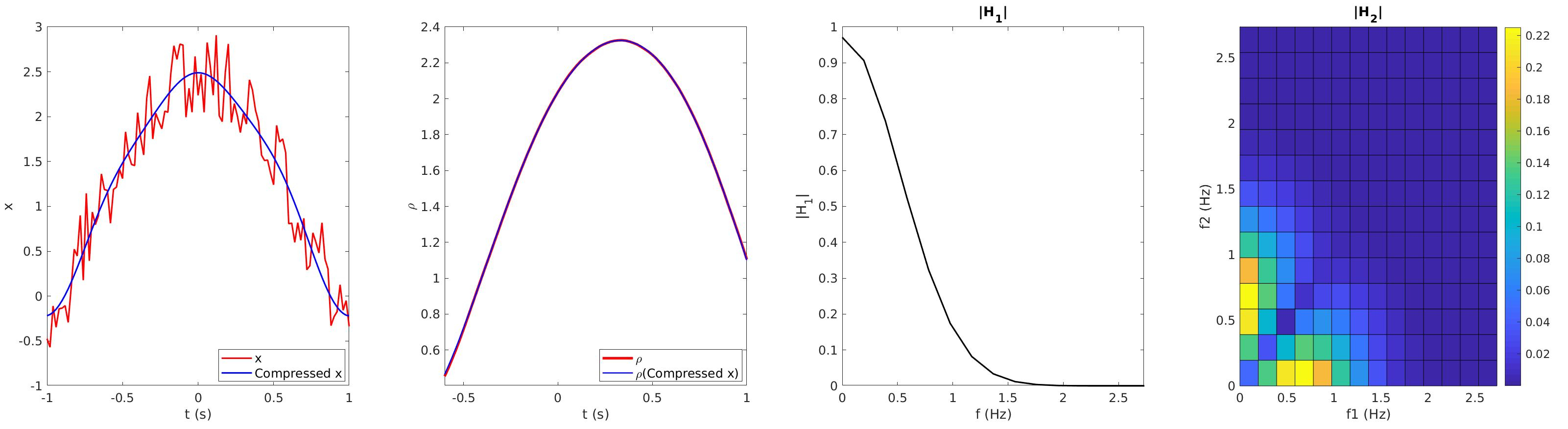}
	\caption{Monitoring-safe compression enabled by Logical Signal Processing. (Left) Signal $\sttraj$ (red) is compressed by eliminating frequencies above the formula's cut-off frequency (1.5Hz). Resulting signal is in blue. (Second panel) Robustness of original and compressed signals relative to a formula $\formula$. Despite the marked difference between the monitored signals, the robustness signals are almost identical, validating the fact that the compression was monitoring-safe. Right two panes show the order-1 and -2 frequency responses of the formula. (Color in digital copy).}
	\label{fig:comp effect once}
\end{figure}

Fig.~\ref{fig:example} shows an example of the results made possible by the methods of this paper, which we refer to as \textit{Logical Signal processing}.
The figure shows the first four Generalized Frequency Response Functions (GFRFs) of the formula $\LTLhist_{[1,1.2]}p$, read `Historically $p$'. Here, $p$ is a given atomic proposition and $\LTLhist$ is the Historically operation, the past equivalent of Always/Globally.
The way to interpret Fig.~\ref{fig:example} is, roughly, as follows: let $X$ be the Fourier transform of the input monitored signal, and $Y$ be the Fourier of the output robustness signal computed by the monitor.
Then $H_1(\omega)$ weighs the contribution of $X(\omega)$ to $Y(\omega)$, $H_2(\omega_1,\omega_2)$ weighs the contribution of the product $X(\omega_1)X(\omega_2)$, $H_3(\omega_1,\omega_2,\omega_3)$ weighs the contribution of  $X(\omega_1)X(\omega_2)X(\omega_3)$, etc.
(Product terms appear because the monitor is a non-linear operator, as will be explained later).
Using these GFRFs $H_n$, we can, for example, calculate frequencies $\omega$ s.t. $X(\omega)$ contributes very little to the output $Y(\omega)$, that is, frequencies of the signal that are irrelevant to the monitor.
This is shown in Fig.~\ref{fig:comp effect once}: using the methods of this paper, we obtained the frequency responses of formula $\LTLonce_{[0.2,0.4]}p$, read `Once $p$'.
From the responses we determined (roughly) that the $H_n$'s are negligible above $\omega=3\pi$ rad/s.
Therefore, setting $X(\omega)$ to 0 above the cut-off should change the monitoring output very little.
This analysis is confirmed by the experiment: Fig.~\ref{fig:comp effect once} shows that there's a marked difference between original and compressed signals, and a priori there is no reason to assume that their robustnesses would be similar. And yet, the calculated robustness signals are almost indistinguishable, confirming that compression was done in a monitoring-safe manner. 
Thus, a codec can suppress all frequencies above the cut-off before transmitting the signal, thus reducing the amount of transmitted bits in a monitoring-safe manner. 

In this paper, our contributions are:
\begin{itemize}
	\item a definition of atomic propositions on the basis of distributions to ensure measurements have finite bandwidth, such that the resulting logic has sound robust semantics. The new logic differs from standard Signal Temporal Logic only in the atomic propositions.
	\item a modeling of temporal logic monitors using Volterra series
	\item the first frequency analysis of temporal logic formulas and their robustness monitors. 
\end{itemize}

The proposed analysis using Volterra series can be used as a signal processing tool for frequency domain analysis of temporal logic monitors but it is not supposed to replace the monitors themselves. 
It can be used for example to design filters which respect the  monitorability requirements. 

\paragraph{Related Work.}
\label{sec:related}
There are few works dealing with the dual aspects of time and frequency in temporal logic.
In~\cite{Donze12TFL} the authors introduce a variant of Signal Temporal Logic (STL), called Time-Frequency Logic (TFL), to explicitly specify properties of the frequency spectrum of the signal - e.g. one formula might read ``$|X(\omega,t)|>5$ over the next three time units'' (where $X(\omega,t)$ is obtained by a windowed Fourier transform and $\omega$ is given).
What we do in this paper is orthogonal to~\cite{Donze12TFL}: given a standard STL formula, we want to analyze which frequencies of signal $x$ contribute to the monitoring verdict. 
We do not require the formula to tell us that explicitly, as in TFL.
In fact, our techniques are applicable to TFL itself: a TFL formula like the one above has a `hidden' frequency aspect, namely the temporal restriction ``in the next three time units''.
Our methods allow an automatic extraction of that.

The works~\cite{Alena16LTI} and~\cite{Ezio18AlgRV} provide a general algebraic framework for the semantics of temporal logic based on semi-rings, opening the way to producing new semantics automatically by concretizing the semi-ring operations.
In the special case of Metric Temporal logic, \cite{Alena16LTI} shows a formal similarity between convolution and the classical semantics of some logical operations. 
However, no frequency analysis is made (and indeed that is not the objective of~\cite{Alena16LTI}). 

In~\cite{Ezio18SCL}, the logic SCL is introduced which replaces the usual temporal operators (like Eventually) with a cross-correlation between the signal and a kernel, which resembles our use of it for atomic propositions (we preserve all temporal operators).
The objective of~\cite{Ezio18SCL} is to measure the fraction of time a property is satisfied, and no frequency analysis is made.
We also mention~\cite{FainekosP09tcs} which derives continuous-time verdicts from discrete-time reasoning.
The modified measurement in this paper should allow us to connect that work to the sampling process and therefore optimize it for monitoring (but we do not carry out that connection here).

Finally, the works~\cite{Abbas18ACCGeneralizedMTL,Abbas19GeneralClasses} consider the impact of representation basis and out-of-order reconstruction on monitoring verdicts, thus also contributing to the study of how signal processing affects runtime verification. 
The work~\cite{Abbas19GeneralClasses} implicitly replaces instantaneous measurements in the temporal operators with a scrambled, or locally orderless, measurement, but maintains instantaneous measurements in atoms. 
No frequency analysis is made.
\section{Preliminaries: Signal Temporal Logic and Volterra Series}
\label{sec:stl and volterra}
\textit{Terminology.} 
Let $\reals=(-\infty,\infty), \reals_+ = [0,\infty), \CoRe = \reals \cup \{\pm \infty\}, \Ne = \{0,1,2,\ldots\}$.
We write $\Ce$ for the set of complex numbers and $i = \sqrt{-1}$.
Given an interval $I \subset \reals$ and $t\in \reals$, $t-I \defeq \{t' \such \exists s \in I . t' = t-s\}$. 
E.g. $t-[a,b] = [t-b, t-a]$.

For a vector $x\in \Re^d$, $|x|$ is its 2-norm and if $x\in\Ce$, $|x|$ is its modulus (i.e. $|a+ib|=\sqrt{a^2+b^2}).$
For a function $f$, $\|f\|$ is its essential supremum, defined as $\|f\| \defeq \inf\{a \such |f(t)|\leq a \text{ on all sets of non-zero measure in the domain of } f\}$.

Given two sets $A$ and $B$, $A^B$ is the set of all functions from $B$ to $A$.
The space of bounded continuous functions $f:\reals^n \rightarrow \reals^m$ is denoted by $\Cc(\reals^n, \reals^m)$.
An $n$-dimensional function is one whose domain is $\reals^n$.
The Fourier transform of function $h$ will be written $\Fourier{h}$. We adopt the convention of using capitalized letters for the transform, e.g. $\Fourier{h} = H, \Fourier{g} = G$, etc.
The magnitude $|H|$ of a Fourier transform will be called the \textit{spectrum}.

Dirac's delta distribution is $\delta$. It is common to abuse notation and treat $\delta$ as a function; in that case recall that for any continuous function $f$ and $\tau \in \reals$, $\int f(t)\delta(t-\tau) dt = f(\tau)$, $\delta(0)=\infty$ and $\delta(t)=0~\forall t\neq 0$.  

In this paper, we reserve the word \textit{operator} to mean a function $\Nlo$ that maps signals to signals, e.g. $\Nlo: \Cc(\reals,\reals) \rightarrow \Cc(\reals,\reals)$. 
The composition of operators $\Nlo$ and $\Mc$ is $\Nlo \circ \Mc$.
A \textit{logical operation}, on the other hand, refers to an operation of a temporal logic, like negation, disjunction, Since and Once. 
Logical True is $\top$, False is $\bot$.

\subsection{Signal Temporal Logic (STL)}
\label{sec:prelim stl}
Signal Temporal Logic (STL)~\cite{MalerN2004STL,Donze10STLRob} is a logic that allows the succinct and unambiguous specification of a wide variety of desired system behaviors over time, such as ``The vehicle reaches its destination within 10 time units while always avoiding obstacles'' and ``While the vehicle is in Zone 1, it must obey that zone's velocity constraints''.

We use a variant of STL which uses past temporal operators instead of future ones. 
For simplicity in this paper we work with scalar-valued signals.
Formally, let $\stSet \subseteq \Re$ be the state-space and let $\TDom \subseteq \reals_+$ be an open interval (which could be all of $\reals_+$).
A \emph{signal} $\sttraj$ is a continuous bounded function.
Let $\{\mu_1,\ldots,\mu_L\}$ be a set of real-valued functions of the state: 
$\mu_k: \stSet \rightarrow \Re$.
Let $AP=\{p_1,\ldots,p_L\}$ be a set of atomic propositions. 

\begin{definition}[Past Bounded STL]
	\label{def:pastBSTL}
	The syntax of the logic \pastBSTL~is given by 
	\[\formula  \defeq \LTLtrue~|~p~|~\neg \formula ~|~ \formula_1 \lor \formula_2 ~|~ \LTLonce_I \formula ~|~ \LTLhistorically_I \formula ~|~ \formula_1 \LTLsince_I \formula_2\]
	where $p \in AP$ and $I  \subset \Re$ is a compact interval.
	The semantics are given relative to signals as follows.
	\begin{eqnarray*}
		\label{eq:boolean sat}
		(\sstraj,t) \models \top &\text{ iff }& \top
		\\
		\forall p_k \in AP, (\sstraj, t) \; \models p_k &\text{ iff }& \mu_k(\sttraj(t)) \geq 0
		\\
		(\sstraj,t) \models \neg \formula&\text{ iff }&  (\sstraj,t) \not\models \formula
		\\
		(\sstraj,t) \models  \formula_1 \lor \formula_2&\text{ iff }& (\sstraj,t) \models \formula_1 \text{ or } (\sstraj,t) \models \formula_2
		\\
		(\sstraj,t) \models \LTLonce_I \formula &\text{ iff }&  \exists t' \in t-I ~.~(\sstraj,t') \models \formula
		\\
		(\sstraj,t) \models \LTLhistorically_I \formula &\text{ iff }&  \forall t' \in t-I~(\sstraj,t') \models \formula
		\\
		(\sstraj,t) \models \formula_1 \LTLsince_I \formula_2 &\text{ iff }& \exists t' \in t-I~.~ (\sstraj,t') \models \formula_2 \text{ and } \forall t'' \in (t',t], \;\; (\sstraj,t'') \models \formula_1 
	\end{eqnarray*}
\end{definition}
It is possible to define the Once ($\LTLonce$) and Historically ($\LTLhistorically$) operations in terms of Since $\LTLsince$, but we make them base operations because we will work extensively with them.

The requirement of bounded intervals and past time is needed to enable the Volterra approximation of temporal logic operations, as will be made explicit in Section~\ref{sec:volterra apx of festl}.

\subsection{Robust Semantics}
\label{sec:prelim robustness}
The robust semantics of an \pastBSTL~formula give a quantitative measure of how well a formula is satisfied by a signal $\sttraj$. 
Usually, robustness is thought of as a functional that maps a signal $\sttraj$ and time instant $t$ to a real value $\robf(\sttraj,t)$ which is the robustness of $\sttraj$ relative to $\formula$ at $t$.
In this paper, instead, we will think of it as an \textit{operator} mapping signals $\sttraj$ to robustness signals $\robf(\sttraj)$. 
This forms the starting point of our Volterra approximation and frequency domain modeling.
\begin{definition}[Robustness\cite{FainekosP09tcs,Donze10STLRob}]
	\label{def:robustness estimate}
	Let $\formula$ be a \pastBSTL~formula.
	The \emph{robustness} $\robf$ of $\formula$ is an operator which maps signals $\sttraj: \TDom \rightarrow \stSet$ to signals $\robf(\sttraj):\reals \rightarrow \CoRe$, and is defined as follows: for any $t \in \reals$,
	\begin{eqnarray*}
		\label{eq:robustness estimate}
		\rob_{\top} (\sstraj)(t) &=& +\infty
		\\
		\rob_{p_k} (\sstraj)(t) &=& \mu_k(\sttraj(t))\,\forall p_k \in AP
		\\
		\rob_{\lnot \formula} (\sstraj)(t) &=& - \rob_{\formula} (\sstraj)(t) 
		\\
		\rob_{ \formula_1 \lor \formula_2} (\sstraj)(t) &=& \max \{\rob_{\formula_1} (\sstraj)(t) , \rob_{\formula_2} (\sstraj)(t)\} 
		\\
		\rob_{\LTLonce_I \formula}(\sstraj)(t) &=&  \max_{t' \in t-I} \robf(\sstraj)(t') 
		\\
		\rob_{\LTLhistorically_I \formula}(\sstraj)(t) &=&  \min_{t' \in t-I} \robf(\sstraj)(t') 
		\\
		\rob_{ \formula_1 \LTLsince_I \formula_2} (\sstraj)(t) &=& \max_{t' \in (t-I)} \left \{\min\{ \rob_{\formula_2} (\sstraj)(t') , \min_{t'' \in (t', t]}   \rob_{\formula_1} (\sstraj)(t'')\} \right\}  
	\end{eqnarray*}
\end{definition}

The following soundness property allows us to use robustness to monitor signals.
\begin{theorem}[Soundness~\cite{FainekosP09tcs}]
	\label{thm:robustness}
	For any signal $\sstraj$ and \pastBSTL~formula $\formula$, 
	if $\robf(\sstraj)(t) <0$ then $\sstraj$ violates $\formula$ at time $t$, and if $\robf(\sstraj)(t) > 0$ then $\sstraj$ satisfies $\formula$ at $t$. 
	Moreover, for every signal $\outtraj$ s.t. $\|\sttraj - \outtraj\|<\robf(\sstraj)(t)$, $(\outtraj,t)\models \formula$ if $(\sttraj,t)\models \formula$ and $(\outtraj,t)\not\models \formula$ if $(\sttraj,t) \not\models \formula$.
\end{theorem} 


\subsection{Fourier Analysis}

We give a brief overview of the Fourier transform and LTI systems; readers familiar with this material can skip this section without loss of continuity.
Fourier analysis allows us to decompose a signal into its constituent frequencies, e.g. by decomposing it into a weighted sum of sinusoids or complex exponentials.
We can then also compute how much energy is placed in a given frequency band. 
  
The Fourier transform $X:\reals \rightarrow \cmplx$ of an input signal $\sttraj$ is defined as:
\begin{equation}
\label{eq:fourier(x(t))}
X(\omega) =  \int_{-\infty}^{\infty}x(t)e^{-i\omega t}dt 
\end{equation}
The real variable $\omega$ is the angular frequency measured in rad/sec and relates to the usual frequency $f$ in Hz as $\omega=2\pi f$.
The magnitude $|X|$ is called the amplitude spectrum of $\sttraj$; the energy in frequency band $[\omega_1, \omega_2]$ is $\int_{\omega_1}^{\omega_2} |X(\omega)|^2 d\omega$.
The Fourier transform is invertible and $\sttraj$ can be obtained using the inverse Fourier transform:
\begin{equation}
\label{eq:inversefourier(X(f))}
x(t) = \frac{1}{2\pi}\int_{-\infty}^{\infty}X(\omega)e^{i\omega t}d\omega
\end{equation}
Thus we can see that $\sttraj$ at $t$ is a weighted sum of complex exponentials, in which $e^{i\omega t}$ is weighted by $X(\omega)$.
For a quick implementation on computers, the discrete version of Fourier transform is evaluated using the Fast Fourier Transform (FFT) algorithm.

The Fourier transform is a powerful tool for studying linear time-invariant (LTI) systems. 
An LTI system is characterized by its \textit{impulse response} $h: \reals \rightarrow \reals$. For an input signal $x$, the system's output signal $y$ is given by the convolution (represented by operator $*$) of $\sttraj$ with the impulse response as follows
\begin{equation}
\label{eq:impulse_response_yt)}
y(t) = (x*h)(t)=  \int_{-\infty}^{\infty}h(\tau)x(t-\tau)d\tau
\end{equation}
The Fourier Transform reduces this convolution to simply the product of the Fourier transforms:
\begin{equation}
\label{eq:impulse_response_yw}
y(t)=(x*h)(t) \leftrightarrow Y(\omega) = X(\omega)H(\omega)
\end{equation}
 
Thus if we choose an LTI system such that $H(\omega)=0$ above some frequency $\omega_c$, we would get $y(t)$ without high frequency noise. Hence, the Fourier domain can be done for designing filters that pass or block specific frequency components of the input signal. 

But Eq.~\eqref{eq:impulse_response_yw} holds only for LTI systems, because complex exponentials are the eigenfunctions for linear, time invariant systems. Since robustness operators used in monitoring temporal logic monitors are nonlinear, they require separate treatment. 
A nonlinear extension is necessary which is provided by Volterra series.

\subsection{Volterra Series Approximation of Non-Linear Operators}
\label{sec:prelim volterra}
A \textit{finite Volterra series operator} $\Nlo$ is one of the form
\begin{equation}
\label{eq:volterra def}
(\Nlo \sttraj)(t) \defeq h_0 + \sum_{n=1}^{N} \int \ldots \int h_n(\tau_1,\ldots,\tau_n)\cdot \sttraj(t-\tau_1) \ldots \sttraj(t-\tau_n) d\tau_1 \ldots d\tau_n
\end{equation}
where $\sttraj$ is the input signal. 
A Volterra series generalizes the convolution description of linear time-invariant (LTI) systems to time-invariant (TI) but nonlinear systems.
We will drop the parentheses to simply write $\Nlo u$ for the output signal of $\Nlo$.
The $n$-dimensional functions $h_n:\reals^n \rightarrow \reals, n\geq 1$, are known as \textit{Volterra kernels}, and their Fourier transforms $H_n:\Ce^n \rightarrow \Ce$ are know as Generalized Frequency Response Functions (GFRFs):
\[H_n(\Omgn)\defeq  \int_{\Taun \in \reals^n}\exp(-i\Omgn^T\Taun)h_n(\Taun)d\Taun\]

We will use Volterra series to approximate the robustness nonlinear operator because
there exists a well-developed theory for studying their output spectra using the GFRFs.
For instance, the Fourier of the output signal $y = \Nlo \sttraj$ is $Y(\omega) = \sum_n Y_n(\omega)$ where ~\cite{Jing15VolterraBook}
\begin{equation}
\label{eq:GFRF def}
Y_n(\omega)= \frac{1}{\sqrt{n}(2\pi)^{n-1}} {\int_{\omega_1+\ldots+\omega_n=\omega}} H_n(\omega_1,\ldots,\omega_{n})X(\omega_1)\ldots X(\omega_{n})d\omega_1\ldots d\omega_{n}
\end{equation}
Eq.~\ref{eq:GFRF def} gives one way to determine which frequencies of signal $\sttraj$ affect the output robustness signal. 
If a frequency $\omega^*$ is s.t. for almost all  $\omega_1,\omega_2,\omega_3,\ldots$, all the following spectra are below some user-set threshold
\begin{equation}
\label{eq:cutoff}
H_1(\omega^*) , H_2(\omega^*,\omega_2) , H_2(\omega_1, \omega^*) , H_3(\omega^*,\omega_2,\omega_3)  , H_3(\omega_1,\omega^*,\omega_3)  , H_3(\omega_1,\omega_2,\omega^*), \text{ etc}
\end{equation}
then $X(\omega^*)$ contributes very little to the formation of the monitoring output, and can be safely discarded.

Volterra series can approximate time-invariant (TI) operators with {fading memory}.
Intuitively, an operator has fading memory if two input signals that are close in the near past, but not necessarily the distant past, yield present outputs that are close. 

\begin{definition}[Fading memory]
	\label{def:fading memory}
	Operator $\Nlo$ has \emph{fading memory} on a subset $K$ of $\Cc(\reals, \reals)$ if there is an increasing function $w:(-\infty,0] \rightarrow (0,1]$, $\lim_{t\rightarrow -\infty}w(t) = 0$ s.t. for each $u \in K$ and $\epsilon >0$ there is a constant $\delta>0$ s.t. 
	\[\forall v \in K,~\sup_{t\leq 0} |u(t)-v(t)|w(t) <\delta \rightarrow |\Nlo u(0)-\Nlo v(0)|<\epsilon \]
	Such a $w$ is called a \emph{weight function} for $\Nlo$.
\end{definition}

\begin{theorem}[Volterra approximation~\cite{Boyd85FMVolterra}]
	\label{thm:volterra apx}
	Let $K_{M_1,M_2} \defeq \{u \in \Cc(\reals, \reals) \such \|u\| \leq M_1, \|u(\cdot -\tau) - u\| \leq M_2\tau~\forall \tau \geq 0 \}$ for some constants $M_1, M_2$, and let $\epsilon>0$. 
	Let $\Rc$ be \emph{any} TI operator with fading memory on $K_{M_1,M_2}$.
	Then there is a finite Volterra series operator $\Nlo$ such that for all $u\in K_{M_1,M_2}$, $\|\Rc u - \Nlo u\| < \epsilon$
\end{theorem}

In practice, how one obtains the Volterra approximation of a given non-linear operator depends on the operator. The probing method~\cite{Jing15VolterraBook} can be used for systems given by ODEs or auto-regressive equations.
However, it is not applicable in our case because it requires feeding complex exponentials to the operator, whereas our robustness operators can only be applied to real-valued signals.
If the operator's behavior is given by a set of input-output pairs of signals, one can first fit a function to the data, then obtain the Volterra representation of that function - see~\cite{Jing15VolterraBook,Boyd85FMVolterra}.

\subsection{Measurement Devices and the spaces $\SigSpace$ and $\FilterSp$}
\label{sec:measurements and atoms}

A measurement device is modeled in classical Physics using the theory of distributions. 
Giving even a cursory overview of this theory is beyond the scope of this paper. 
We will provide the necessary mathematical definitions and refer the reader to~\cite{FlorackIP} for a more comprehensive CS and Engineering-oriented introduction to this topic.

Let $\SigSpace$ be the space of infinitely differentiable functions with compact support in $\TDom$.
A \textit{measurement kernel} in this paper is a non-zero function $f: \reals \rightarrow \reals$ with $L_1$ norm at most 1, i.e., $\|f\|_1 \defeq \int |f(t)|dt \leq 1$.  
Let $\FilterSp$ be the space of all such functions.
Note that $f \in \FilterSp$ iff $-f \in \FilterSp$ and that the shifted kernel $f(\cdot -t)$ is in $\SigSpace$ for every $t$ whenever $f\in \SigSpace$. 
The measurement signal $y$ is then obtained by taking the following inner product:
\begin{equation}
\label{eq:y(t)}
y(t) = \langle f(\cdot -t),\sttraj \rangle \defeq \int_{-\infty}^{\infty}f(\tau-t)\sttraj(\tau)d\tau\quad \forall t
\end{equation}
One can think of the measurement device as taking an $f$-weighted average of the values of $\sttraj$ centered on $t$.
Informally, the width of $f$ dictates the resolution of the measurement: the narrower $f$, the higher the resolution.
Different measurement devices use different filters $f \in \FilterSp$.
Note that Dirac's $\delta$ is not in $\FilterSp$, but can be approximated arbitrarily well with narrow integrable functions.

\section{Bounded-Bandwidth STL}
\label{sec:bbstl}

The semantics of an atomic proposition, given in Def.~\ref{def:pastBSTL}, requires the ability to measure the instantaneous value $\sstraj(t)$.
However, it is \textit{physically impossible to do an instantaneous measurement}~\cite{FlorackIP}: in \eqref{eq:y(t)}, $y(t) = \sttraj(t)$ iff $f = \delta$, Dirac's delta. 
But $\delta$ is not realizable because it has infinite energy: $\int \delta^2(t)dt = \delta(0) = \infty$.
In this paper, we must pay closer attention to how measurements are actually made in the physical world for two reasons:

\begin{itemize}
	\item we are interested in analyzing runtime monitors when they are a part of a signal processing chain.  
	If something is not physically possible, e.g., because it requires infinite energy, it makes little sense to model how other components in the chain will process its output.	
	\item We are interested in analyzing the input-output relation of a temporal logic monitor in the frequency domain (its transfer function, as it were). Even if we kept using instantaneous measurements in the theory for convenience sake, we'd end up with the trivial result that all robustness monitors have infinite essential bandwidth~\cite{Slepian76OnBandwidth} since $\Fourier{\delta}(\omega) =  1~\forall \omega$. 
	I.e., \textit{all} frequency bands are relevant - clearly a useless analysis.
\end{itemize}

This motivates our introduction of a new logic, Bounded-Bandwidth STL (\festl, pronounced `baby Steel'),  that does away with punctual measurements while preserving the logical relations and the soundness of robust semantics.

\festl~formulas are interpreted over signals in $\SigSpace$, defined in Section~\ref{sec:measurements and atoms}. 
Let $AP$ be a set of atomic propositions s.t. there exists a bijection between $\FilterSp$ and $AP$. 
We write $p_f$ for the atom corresponding to filter $f$.
\begin{definition}[Bounded-Bandwidth STL]
	\label{def:fenstl}
The syntax of \festl~is identical to that of \pastBSTL:
\[\formula  \defeq \LTLtrue~|~p_f~|~\neg \formula ~|~ \formula_1 \land \formula_2 ~|~ \LTLonce_I \formula ~|~ \LTLhistorically_I \formula ~|~ \formula_1 \LTLsince_I \formula_2\]
where $p_f \in AP$ and $I  \subset \Re$ is a compact interval.
Its boolean semantics are identical to those of \pastBSTL~except for the atomic proposition case given here:
\[(\sttraj,t) \models p_f \text{  iff  } \langle f(\cdot -t), \sttraj\rangle \geq 0\]
Its robust semantics are identical to those of \pastBSTL~except for the base case below. 
\begin{equation*}
	\label{eq:festl robustness}
	\rob_{p_f} (\sstraj)(t) = \langle f(\cdot -t), \sttraj\rangle
\end{equation*}
\end{definition}
The robustness of any signal relative to any atomic proposition is finite: letting $S_\sttraj$ be the compact support of signal $\sttraj$, it holds that $\langle f, \sttraj\rangle \leq \int_{S_\sttraj}|f|dt \cdot \int_{S_\sttraj} |\sttraj|dt$, 
which is finite since $f$ is absolutely integrable and $\sttraj$ is continuous and therefore bounded on any compact set.
Thus $\rob_{\formula}(\sttraj) \leq \rob_{\top}(\sttraj)$ for any $\formula$, as required for an intuitive interpretation of robustness.

The following theorem establishes that \festl~can be monitored via its robust semantics.
\begin{theorem}[Soundness of robust semantics]
	\label{thm:festl robust monitoring}
	For every signal $\sstraj \in \SigSpace$ and \festl~formula $\formula$, 
	if $\robf(\sstraj)(t) <0$ then $\sstraj$ violates $\formula$ at time $t$, and if $\robf(\sstraj)(t) > 0$ then $\sstraj$ satisfies $\formula$ at $t$. Moreover, for every signal $\outtraj$ s.t. $d(\sttraj,\outtraj)<\robf(\sstraj)(t)$, $(\outtraj,t)\models \formula$ if $(\sttraj,t)\models \formula$ and $(\outtraj,t)\not\models \formula$ if $(\sttraj,t) \not\models \formula$.
\end{theorem} 

Before proving the theorem, we make several remarks about the definition of \festl~and the various restrictions we placed on the signal and kernel spaces.
The measurement process $\sttraj \rightarrow \langle \sttraj, f(\cdot -t)\rangle$ can be written as a convolution $(\sttraj * f^-)(t)$, where $f^-(t) = f(-t)$. 
So $\Fourier{f^-}$ is the transfer function of the measurement process.
By selecting an appropriate set of filters, we get rid of the problem of infinite bandwidth measurements.
In particular, we make sure that $\delta$ is not in $\FilterSp$.

STL and \pastBSTL~use arbitrary functions $\mu_k$ in their atoms, which allows arbitrary processing of the signal. 
E.g. if some $\sttraj$ is 1-D, and we want to express the requirement $x^2 - e^x \geq 0 \land x\geq 1$, we can do that by using $\mu_1(x) = x^2 - e^x$ and $\mu_2(x) = x-1$.
\festl~does not have that expressiveness, but we are nevertheless able to compute arbitrary linear functionals of $\sttraj$ and compare them. 
E.g. the requirement $\langle x,f\rangle \geq 2 \langle x,g\rangle$ is captured as $\langle x, f-2g\rangle \geq 0$.
So the difference between STL and \festl, at the level of atomic propositions, is in the ability to generate auxiliary signals in a non-linear vs linear fashion.

The Volterra approximation of an operator requires the latter to be causal and have fading memory (causality is implied by the conditions of Thm.~\ref{thm:volterra apx}~\cite{Boyd85FMVolterra}).
Causality requires working with past time operations, and fading memory requires working with bounded temporal operators.
This is why we derived \festl~from \pastBSTL~rather than STL. 

To prove Thm.~\ref{thm:festl robust monitoring}, we will first need to define a distance function $d: \SigSpace \times \SigSpace \rightarrow \reals$:
\begin{equation}
\label{eq:distance fnt}
d(\sttraj, \outtraj) \defeq \sup \{\langle \sttraj -\outtraj, f \rangle  \such f\in \FilterSp \}
\end{equation}

\begin{lemma}
	\label{lemma:distance fnt}
	Function $d$ is a metric on $\SigSpace$.
\end{lemma}
\begin{proof}
	$d$ is non-negative: indeed for all $\sttraj \in \SigSpace$ and $g \in \FilterSp$, $\sup_f \langle \sttraj, f\rangle \geq \max(\langle \sttraj, g\rangle, \langle \sttraj, -g\rangle) = |\langle \sttraj, g\rangle|$. Since $\sttraj - \outtraj \in \SigSpace$ whenever $\sttraj,\outtraj \in \SigSpace$, the conclusion follows.
	
	$d$ is symmetric: $d(\sttraj,\outtraj) = \sup_f \langle \sttraj - \outtraj, f\rangle = \sup_{f} \langle  \outtraj - \sttraj, -f\rangle =\sup_{f \in -\FilterSp} \langle  \outtraj - \sttraj, f\rangle = \sup_{f \in \FilterSp} \langle  \outtraj - \sttraj, f\rangle = d(\outtraj,\sttraj)$.
	
	$d$ satisfies the triangle inequality: for any $\sttraj, \outtraj, \ztraj \in \SigSpace$,
	\begin{equation*}
	d(\sttraj,\outtraj) = \sup\{\langle \sttraj + \ztraj-\ztraj -\outtraj , f\rangle \such f\in \FilterSp\} 
	\leq \sup_{\FilterSp}\{\langle \sttraj - \ztraj , f\rangle\} + \sup_{\FilterSp} \{\langle \ztraj -\outtraj , f\rangle \}
	= d(\sttraj,\ztraj) + d(\ztraj,\outtraj)
	\end{equation*} 
	
	$d$ separates points: that is, if $d(\sttraj,\outtraj)=0$ then $\sttraj=\outtraj$.
	We will argue by contradiction.
	Define function $\varepsilon$ by $\varepsilon(t)= \sttraj(t)-\outtraj(t)$.
	Assume $\sttraj\neq \outtraj$ so there exists a $t'\in \TDom$ s.t. $\varepsilon(t') \neq 0$ and without loss of generality we may assume $\varepsilon(t') >0$ (since $-\varepsilon \in \SigSpace$) and that $t'=0$.
	Since $\varepsilon$ is continuous, there exists a neighborhood $I$ of 0 over which $\varepsilon(t)>0$. 
	So pick $g\in \FilterSp$ s.t. $g(t)>0$ over $I$ and 0 elsewhere.
	It then holds that $\langle g,\varepsilon\rangle >0$, contradicting $d(\sttraj,\outtraj)=0$. 
	Therefore $\varepsilon=0$ and $\sttraj = \outtraj$.
	\IEEEQED
\end{proof}

Metric $d$ takes the distance between signals to be the largest measurement that can be made of their difference; this is consistent with the view that \textit{we have no access to a signal without a measurement device}. 
The only way to differentiate between signals is to measure a difference between them.
(Eq.~(2.6) in~\cite{BlanchardMathMethods} gives a more widely applicable metric but $d$ above is much more interpretable).
We can now proceed with the proof of Thm.~\ref{thm:festl robust monitoring}.
\begin{proof}
Let $\Lang$ be the set of all $\sttraj$ s.t. $(\sttraj,t)\models \formula$ and for a subset $S \subset \SigSpace$ let $\dist(\sttraj,S) = \inf_{\outtraj\in S} d(\sttraj,\outtraj)$.
By convention set $\dist(\sttraj,\emptyset) = \infty$.
Following~\cite{FainekosP09tcs}, and given that $d$ is a metric, it suffices to show that the following inequality holds for the base cases $\formula=\top$ and $\formula=p_f$:
\[-\dist(\sttraj, \Lang) \leq \robf(\sttraj)(t) \leq  \dist(\sttraj, \SigSpace \setminus \Lang)\]
The remaining cases then follow by structural induction on $\formula$.

\underline{$\formula=\top$} Then $\sttraj \in \Lang$ for any $\sttraj$ and so $\dist(\sttraj, \Lang) = 0 \leq \infty = \robf(\sttraj)(t) =\dist(\sttraj,\emptyset)= \dist(\sttraj, \SigSpace \setminus \Lang) $.

\underline{$\formula=p_f$}. Suppose $\sttraj \in \Lang$. For all $\outtraj \in \SigSpace \setminus \Lang$
\begin{eqnarray*}
d(\sttraj, \outtraj)&\geq & \langle \sttraj-\outtraj, f(\cdot -t)\rangle \text{ since } f(\cdot -t ) \in \FilterSp
\\
&=&  \langle \sttraj,f(\cdot -t)\rangle -\langle \outtraj, f(\cdot -t)\rangle
\\
&\geq&\langle \sttraj,f(\cdot -t)\rangle \text{ since } \langle \outtraj, f(\cdot -t)\rangle < 0
\\
&=& \rob_{p_f}(\sttraj)(t)
\end{eqnarray*}
Thus, $\dist(\sttraj, \Lang) = 0 \leq \robf(\sttraj)(t) \leq \dist(\sttraj, \SigSpace \setminus \Lang)$.

Now suppose $\sttraj \notin \Lang$.
As before $\inf_{\outtraj \in \Lang}d(\sttraj,\outtraj)\geq \inf_{\outtraj \in \Lang} \langle \outtraj,f(\cdot -t)\rangle -\langle \sttraj, f(\cdot -t)\rangle$ 
so $\dist(\sttraj,\Lang) \geq -\langle \sttraj, f(\cdot -t)\rangle$.
Thus, $- \dist(\sttraj,\Lang) \leq \langle \sttraj, f(\cdot -t)\rangle = \rob_{p_f}(\sttraj)(t) < 0 = \dist(\sttraj, \SigSpace \setminus \Lang)$.
 \IEEEQED
\end{proof}

\section{Volterra Approximations and Frequency Response of \festl~Formulas}
\label{sec:volterra apx of festl}
\newcommand{\Nhat}{\widehat{\Nlo}}

Having defined the logic \festl, we are now in a position to answer the question: what is the frequency content of temporal logic?
The strategy will be to show that the robustness of each logical operation ($p_f, \neg, \lor, \LTLonce_I, \LTLpastglobally_I, \LTLsince_I$) can be approximated by a Volterra series, and derive its GFRF.
Then using a composition theorem, we can derive the GFRF of entire formulas to deduce which frequencies are given significant weight by the GFRF, and which aren't.

We note at the outset that the robustness operator for $\top$, $\rob_{\top}$, maps any signal to the improper constant function $t \mapsto +\infty$. 
Because this function is not in $\Cc(\reals,\reals)$, $\rob_{\top}$ is not approximable by a finite Volterra series on the basis of Thm.~\ref{thm:volterra apx}.
This is not a serious impediment, since it is highly unlikely that an engineer would explicitly include $\top$ in a specification (e.g. $\formula = p \lor \top$), so there is no need to approximate $\rob_\top$ to begin with. 
As for formulas that accidentally turn out to be tautologies, like $p\lor \neg p$, their STL robustness is not infinite, and neither is their \festl~robustness. 

\subsection{Approximability by Volterra Series}
\label{sec:volterra ops}
We state and prove the main technical result of this paper.
\begin{theorem}
	\label{lemma:operator has volterra apx}
	For any \festl~formula $\formula$ that does not explicitly include $\top$, the robustness operator $\robf:\SigSpace \rightarrow \CoRe^\reals$ can be approximated by a finite Volterra series.
\end{theorem}

	Recall the set $K_{M_1,M_2}$ from Thm.~\ref{thm:volterra apx}, and recall that for a function $f$, $\|f\|$ is its essential supremum.	
	We will first show that $\robf$ is TI and has fading memory.
	However, the domain of $\robf$ is not a set of the form $K_{M_1,M_2}$ so we can't apply Thm.~\ref{thm:volterra apx} directly.
	So we show how to roughly decompose $\SigSpace$ into sets of the form $K_{M_1,M_2}$ and leverage Thm.~\ref{thm:volterra apx} to conclude.
	In all that follows, it is understood that $\formula$ does not explicitly include $\top$.
	
	\begin{lemma}
		\label{lemma:TI FM rob}
		The operator $\robf$ is TI and has fading memory.
	\end{lemma}
	\begin{proof}
				Time invariance is immediate.
				To prove fading memory we argue by induction on the structure of $\formula$.
				
				\textbf{Base cases.} Fix an arbitrary $p_f$.
				We must exhibit a weight function s.t. for all $\varepsilon>0$ and $u,v\in \SigSpace$,
				$\sup_{t'\leq 0}|u(t)-v(t)|w(t)<\delta \implies |N_fu(0)-N_fv(0)| = |\int f(\tau)(u(\tau)-v(\tau))d\tau| <\varepsilon$.
				Fix $\varepsilon>0$, and let $w$ be a continuous increasing function from $(-\infty,0]$ to $(0,1]$.
				For every $u,v\in \SigSpace$, $g\defeq u-v$ is in $\SigSpace$; let $C$ be its compact support.
				If $\sup_{t'\leq 0}|g(t')|w(t')<\delta$ then 
				\[\left| \int f(t)g(t)dt \right| = \left|\int_C f(t)g(t)dt \right| \leq \int_{C\cap (-\infty, 0]} |f(t)| |g(t)|dt < \delta \int_{C\cap (-\infty, 0]} |f(t)|/w(t)dt\]
				The integral is finite and non-zero so choosing $\delta=\varepsilon/(\int |f(t)|/w(t)dt)$ yields the desired result.
				
				\textbf{Inductive cases.} 
				The case of $\neg \formula$ is immediate.
				
				$\bullet$ For $\formula_1 \lor \formula_2$: by the induction hypothesis there exist weight function $w_k$ for $\rob_{\formula_k}$, $k=1,2$ s.t. for all $\varepsilon>0$, $\sup_{t\leq 0}|u(t)-v(t)|w_k(t)<\delta \implies |\rob_{ \formula_k}(u)(0)- \rob_{ \formula_k}(u)(0)|<\varepsilon$. 
				Then $w= \max\{w_1,w_2\}$ is easily shown to be a weight function for $\rob_{\formula_1 \lor \formula_2}$.
				
				$\bullet$ For $\LTLonce_I \formula$: 
				By the induction hypothesis, there exists a weight function $w$ s.t. for all $\varepsilon>0$ and $u,v\in \SigSpace$ there exists $\delta>0$ s.t.
				\begin{equation}
				\label{eq:bla}
				\sup_{t\leq 0}|u(t)-v(t)|w(t)<\delta \implies |\robf(u)(0)-\robf(v)(0)|<\varepsilon/2
				\end{equation}
				\qquad \textit{Fact.} If $\sup_{t\leq 0}|u(t)-v(t)|w(t)<\delta $ then 
				$\sup_{\tau\leq 0}|\robf (u)(\tau)-\robf (v)(\tau)|\leq \varepsilon/2$.
				
				Indeed,  if $\sup_{t\leq 0}|u(t)-v(t)|w(t)<\delta $ then it holds that for all $\tau\geq 0$, $\sup_{t\leq -\tau}|u(t)-v(t)|w(t)<\delta$, which is equivalent (by a change of variables) to $\sup_{t\leq 0}|u(t-\tau )-v(t-\tau)|w(t-\tau)<\delta$.
				But $w(\cdot -\tau)\leq w$ so 
				\[\sup_{t\leq 0}|u(t-\tau )-v(t-\tau)|w(t-\tau)<\delta \implies \sup_{t\leq 0}|u(t-\tau )-v(t-\tau)|w(t)<\delta\]
				Since $u(\cdot-\tau), v(\cdot-\tau)$ are in $\SigSpace$ it follows that $|\robf (u)(-\tau)-\robf (v)(-\tau)|<\varepsilon/2$ for all $\tau\geq 0$, and therefore $\sup_{\tau\leq 0}|\robf (u)(\tau)-\robf (v)(\tau)|\leq \varepsilon/2$.
				
				Now we claim that $w$ is a weight function for $\rob_{\LTLonce_I \formula}$. 
				Indeed  $|\rob_{\LTLonce_I \formula}(u)(0) - \rob_{\LTLonce_I \formula}(v)(0)| = |\max_{t\in -I}\robf(u)(t)-\max_{t\in -I}\robf(v)(t)|$. 
				Let $t_u = \argmax_{-I}\robf(u)(t)$ and $t_v = \argmax_{-I}\robf(v)(t)$; both exist since $I$ is compact and $\robf$ is continuous.
				Assume the left-hand side of Eq.~\eqref{eq:bla} holds.
				Then we finally find the string of inequalities
				\[-\varepsilon <\varepsilon/2 \leq \robf(u)(t_v) - \robf(v)(t_v) \leq \max_{t\in -I}\robf(u)(t)-\max_{t\in -I}\robf(v)(t) \leq \robf(u)(t_u) - \robf(v)(t_u) \leq \varepsilon/2 < \varepsilon\]
				Therefore $|\robf(u)(0)-\robf(v)(0)|<\varepsilon$ as desired.
				
				$\bullet$ The case of $\LTLhist_I\formula$ is similar.
				
				$\bullet$ For $\psi = \formula_1 \LTLsince_I \formula_2$: suppose there exist weight functions $w_u$ and $w_v$ for $u$ and $v$ respectively.
				Write $\rob_k = \rob_{\formula_k}$, $k=1,2$. 
				Set $w = \max\{w_u,w_v\}$: this will be the weight function for $\rob_\psi$. 
				Given $\varepsilon>0$, there exists a $\delta>0$ s.t.
				$\sup_{t\leq 0}|u(t)-v(t)|w(t) < \delta \implies |\rob_k u(0)-\rob_k v(0)|<\varepsilon$. 
				By the above Fact, it also follows that 
				\begin{equation}
				\label{eq:robk}
				|\rob_k u(t') - \rob_k v(t')|<\varepsilon~\forall t'\leq 0, k=1,2
				\end{equation}
				We will show that $|\rob_\psi u(0)-\rob_\psi v(0)|<\varepsilon$,	
				where $\rob_\psi u(0) = \max_{t'\in -I}\{\min\{\rob_2u(t'), \min_{t''\in (t',0]}\rob_1u(t'')\}\}$. 
				Given $t'\leq 0$, define $t_u \defeq \argmin_{(t',0]} \rob_1 u(t''), t_v \defeq \argmin_{(t',0]} \rob_1 v(t'')$.
				The following inequalities are immediate:
				\[\rob_1v(t_v)-\varepsilon \leq \rob_1v (t_u)-\varepsilon < \rob_1 u(t_u) \leq \rob_1 u(t_v) < \rob_1 v(t_v)+\varepsilon\]
				Therefore 
				\begin{equation}
				\label{eq:rob1}
				|\rob_1 u(t_u) - \rob_1 v(t_v)| <\varepsilon
				\end{equation}
				From Eqs.~\eqref{eq:rob1} and~\eqref{eq:robk} it follows that 
				\[\forall t' \in -I,~|\underbrace{\min\{\rob_2u(t'), \min_{t''\in (t',0]}\rob_1u(t'')\}}_{a(t')} - \underbrace{\min\{\rob_2v(t'), \min_{t''\in (t',0]}\rob_1v(t'')\}}_{b(t')}|<\varepsilon\]
				With $t_a\defeq \argmax_{t'\in -I}a(t'), t_b\defeq \argmax_{t'\in -I}b(t')$
				\[b(t_b)-\varepsilon \leq b(t_a)-\varepsilon < a(t_a) \leq a(t_b) < b(t_b)+\varepsilon\]
				and we get the desired conclusion: $|a(t_a)-b(t_b)| = |\rob_\psi u(0) - \rob_\psi v(0)| <\varepsilon$.
				\IEEEQED
	\end{proof}
	
	We continue with the main proof.
	A signal $\sttraj$ in $\SigSpace$ is infinitely differentiable and compactly supported, so there exist $M_1$ and $M_2$ s.t. $\sttraj \in K_{M_1,M_2}$.
	Moreover for every $M_1' \geq M_1$ and $M_2' \geq M_2$, $K_{M_1,M_2} \subseteq K_{M_1',M_2'}$.
	Thus if we take any ascending sequence $(M_{1,a}, M_{2,a})_{a \in \Ne}$ with first element $(0,0)$ and which is unbounded in both dimensions, we have that $\SigSpace \subset \cup_{a\in \Ne} K_{M_{1,a}, M_{2,a}}$.
	(The lexicographic order is used: $(M_{1,a}, M_{2,a}) \leq (M_{1,a}, M_{2,a})$ iff $M_{1,a} \leq M_{1,a'}$ and $M_{2,a} \leq M_{2,a'}$).
	For conciseness write $K_{a}\defeq K_{M_{1,a}, M_{2,a}}$.
	\begin{lemma}
		\label{lemma:restrict rho to Ka}
		The restriction of $\robf$ to any $K_a$ is an operator over $K_a$, i.e. $\robf:K_a \rightarrow K_a$.
	\end{lemma}
	\begin{proof}
				Take $\sttraj \in K_a$, we show that $y = \robf(\sttraj) \in K_a$.
				For any kernel $f$ and $t \in \reals$, H{\"o}lder's inequality gives $\langle f(\cdot -t ),\sttraj \rangle \leq \int |f(\tau)|d\tau \cdot \|\sttraj\| \leq M_{1,a}$, so $\|\rob_{p_f}(\sttraj) \| \leq M_{1,a}$. 
				Since the robustness of any formula other than $\top$ is obtained by taking max and min of atomic robustness values, $\| \robf(\sttraj) \| \leq M_{1,a}$.
				Moreover for all $t,s \in \TDom$
				\[|y(t)-y(s)| = \left|\int f(\tau) [\sttraj(\tau+t) - \sttraj(\tau+s)]d\tau\right| \leq \|f\|_1\cdot  \|\sttraj(\cdot +t)-\sttraj(\cdot +s)\| \leq M_{2,a} |t-s|\]
				This shows that $y \in K_a$.
				\IEEEQED
	\end{proof}
	
	\begin{lemma}
		\label{lemma:restrict N to Ka}
		Consider an operator $\Nlo : \SigSpace \rightarrow \reals^{\CoRe}$ such that its restriction $\Nlo_a$ to $K_a$ is an operator over $K_a$.
		If $\Nlo$ is TI and with fading memory, then it has a finite Volterra series approximation $\Nhat$ over $\SigSpace$.
	\end{lemma}
	\begin{proof}
	It is immediate that if $\Nlo$ is TI and with fading memory, then so is every $\Nlo_a$.
	Thus, fixing $\varepsilon > 0$, $\Nlo_a$ has a finite Volterra series approximation over $K_a$ by Thm.~\ref{thm:volterra apx}, call it $\Nhat_a$, so that for all $\sttraj \in K_a$, $\|\Nhat_a \sttraj - \Nlo_a \sttraj\| < \varepsilon$.	
	
	For every signal $\sttraj \in \SigSpace$, let $\sttraj'$ be its time derivative. 
	Then $\sttraj \in K_{\|\sttraj\|, \|\sttraj '\|}$, and for all $M_1' < M_1$ and $M_2' < M_2$, $\sttraj \notin K_{M_1',M_2'}$.
	(The first part of the last statement is immediate; for the second part, note first that there exists $t^*$ in the support of $\sttraj$ s.t. $M_2 = |\sttraj'(t^*)|$, so pick $b,c$ s.t. $b\leq t^*\leq c$ and $\sttraj(c) = \sttraj(b) + \sttraj'(t^*)(c-b)$, or $|\sttraj(c) - \sttraj(b)| =  |\sttraj'(t^*)|(c-b) > M_2'(c-b)$).
	So there exists a unique smallest pair $(M_{1,a},M_{2,a})$ s.t. $x\in K_a$, namely the smallest pair s.t. $M_{1,a}\geq \|\sttraj\|$ and $M_{2,a} \geq \|\sttraj'\|$.
	For a given $\sttraj$ let $a(\sttraj)$ be the index in $\Ne$ corresponding to this smallest pair.
	
	Define the operator $\Nhat:\SigSpace \rightarrow \reals^{\CoRe}$ by $\Nhat\sttraj \defeq \Nhat_{a(\sttraj)} \sttraj$.	
	Then for all $\sttraj \in \SigSpace$,  $\|\Nhat\sttraj - \Nlo \sttraj\| = \|\Nhat_{a(\sttraj)} \sttraj - \Nlo_{a(\sttraj)} \sttraj\|  < \varepsilon$, which establishes that $\Nhat$ is a finite Volterra approximation of $\Nlo$ over $\SigSpace$.	
\IEEEQED
\end{proof}

Combining the three lemmas allows us to conclude the main proof.
Even though it is only strictly correct to speak of the Volterra kernels of the Volterra series that approximates the robustness operator $\rob_{\formula}$, we will often abuse language and speak directly of the `Volterra kernels of $\formula$'.

\subsection{Calculating the Volterra Approximations and their GFRFs}
\label{sec:volterra formulas}
We seek the Volterra series that approximates $\robf$ for a given formula in the sense of Thm.~\ref{thm:volterra apx}.
Operator $\robf$ is built by composing a few basic operators.
The strategy will be to first approximate each basic operator by a Volterra series, then use a composition theorem to compose these into a Volterra series for the entire formula.
We exclude the Since operation from the remainder of this discussion because, even though its robustness is approximable by Lemma~\ref{lemma:operator has volterra apx}, we don't currently have the tools to compute that approximation. 
We expand on the technical difficulty of performing that approximation in Section~\ref{sec:apx since}.

\textbf{Basic operators.}
Fix an interval $[a,b] \subset \reals_+$, $\varepsilon >0$ and $f\in \FilterSp$, let $u,v$ denote arbitrary signals in $\SigSpace$.
The basic operators are:
\begin{multicols}{3}
	\noindent
	\begin{equation*}
		\Nlo_{f}u(t) = \langle f(\cdot -t),u \rangle
	\end{equation*}
	\begin{eqnarray*}
			\Nneg u(t) &=& -u(t)		\\
			\Nmax (v,u)(t) &=& \max\{v(t),u(t)\}	
	\end{eqnarray*}
	\begin{eqnarray}
	\label{eq:basic ops}
		\Nhist_{[a,b]} u(t) &=& \min_{t-b\leq t'\leq t-a}u(t') \nonumber \\
		\Nonce_{[a,b]} u(t) &=& \max_{t-b\leq t'\leq t-a}u(t') 
	\end{eqnarray}
\end{multicols}
The following relations hold:
\begin{multicols}{3}
	\noindent
	\begin{equation*}
	\rob_{p_f} = \Nlo_{f} 
	\end{equation*}
	\begin{eqnarray*}
	\rob_{\neg \formula}  &=& \Nneg \circ \robf \\
	\rob_{\formula_1 \land \formula_2} &=&  \Nmin(\rob_{\formula_1}, \rob_{\formula_2}) 
	\end{eqnarray*}
	\begin{eqnarray*}
		\rob_{\LTLhistorically_{[a,b]} \formula}  &=&  \Nhist_{[a,b]} \circ  \robf  \\
		\rob_{ \LTLonce_{[a,b]}\formula} &=&  \Nonce_{[a,b]} \circ  \robf  
	\end{eqnarray*}
\end{multicols}

We approximate each basic operator, on a representative set of signals, using a structure made of delays followed by a read-out polynomial; this structure can then be represented exactly with Volterra series.
It is shown in~\cite{Boyd85FMVolterra} that this structure (delays followed by polynomial) can approximate any discrete-time operator and is a special case of a structure for approximating any continuous-time operator on $\Cc(\reals, \reals)$.

There are many ways to derive Volterra approximations. 
Here we give a practical and simple way of computing such an approximation numerically.
The first two operators can be represented exactly as Volterra series.


$\bullet\underline{\Nlo_{f}u(t) = \langle f(\cdot -t),u \rangle}$.
Then $h_0=0, h_1(t) = f(-t), h_{n} \equiv 0$ when $n>1$.

$\bullet\underline{\Nneg u(t) = -u(t)}$.
Then $h_0=0, h_1(t) = -\delta(t), h_{n} \equiv 0$ when $n>1$.
Note that $\Nneg$ is never applied directly to a source signal (i.e. monitored signal $\sttraj$) but only to robustness signals.
Robustness signals are produced by previous monitors and their values are stored (perhaps symbolically) in computer memory, so it is possible to access their instantaneous values.
So this does not contradict our earlier point about the inability to instantaneously sample an analog \textit{source signal}.

$\bullet\underline{\Nonce_{[a,b]} u}$.
We approximate this operator by a polynomial $P(u(t-t_1), \ldots,u(t-t_D))$ for a given choice of polynomial degree $d$ and delays $t_j$, $a\leq t_j \leq b$ . 
$P$ is of the form $\sum_{\vec{r}} \alpha_{\vec{r}} u(t-t_1)^{r_1}\ldots u(t-t_D)^{r_D}$ where the sum is over all integer vectors $\vec{r} = (r_1,\ldots, r_D)$ s.t. $0\leq r_j\leq d, \sum_jr_j \leq d$, and the $\alpha_{\vec{r}}$'s are the unknown polynomial coefficients. 
Then given a set of $L$ signals $u_\ell$ and the corresponding output signals $\Nonce_{[a,b]}u_\ell$, and given a set $\Te$ of sampling times, we setup the linear system in the $\alpha_{\vec{r}}$'s:
\begin{equation}
\label{eq:lin sys}
\sum_{\vec{r}} \alpha_{\vec{r}} u_\ell(t-t_1)^{r_1}\ldots u_\ell(t-t_D)^{r_D} = \Nonce_{[a,b]}u_\ell(t), 1\leq \ell \leq L, t \in \Te
\end{equation}
A least-squares solution yields the $\alpha$'s.
We force $\alpha_{\vec{0}}=0$ since the $\Nonce$ operator has 0 response to 0 input.
Therefore $h_0 =0$.
Given this approximation we seek the kernels $h_n$ s.t.
\[P(u(t-t_1), \ldots,u(t-t_D)) = \sum_{\vec{r}} \alpha_{\vec{r}} u(t-t_1)^{r_1}\ldots u(t-t_D)^{r_D} = \sum_{n=1}^N \int_{\vec{\tau}\in \reals^n} h_n(\vec{\tau})\prod_{j=1}^{n}u(t-t_j)d\vec{\tau}\]

Define $\Delta_d^D(n)= \{\vec{r} = (r_1,\ldots,r_D) \in \Ne^D \such 0\leq r_j\leq d, \sum_j r_j =n \}$ and let $\Delta_d^D = \cup_{0\leq n \leq d} \Delta_d^N(n)$.
For a given $\vec{r} \in \Delta_d^D(n)$,
\begin{eqnarray*}
&&u(t-t_1)^{r_1}\ldots u(t-t_D)^{r_D} 
\\ 
&&= \int \underbrace{\delta(\tau_1-t_1)\ldots \delta(\tau_{r_1}-t_1)}_{r_1 \text{terms}}    \underbrace{\delta(\tau_{r_1+1}-t_2)\ldots \delta(\tau_{r_1+r_2}-t_2)}_{r_2 \text{terms}}    \ldots \underbrace{\delta(\tau_{n-r_D+1}-t_D)\ldots \delta(\tau_{n}-t_D)}_{r_D \text{terms}} \prod_{j=1}^nu(t-\tau_j)d\vec{\tau}
\end{eqnarray*}

Therefore define $h_n^{\vec{r}}(\tau_1,\ldots,\tau_n) \defeq \alpha_{\vec{r}} \prod_{j=1}^{r_1}\delta(\tau_j-t_1)\ldots \prod_{j=n-r_D+1}^{n}\delta(\tau_j-t_D)$.
We can now express
\begin{eqnarray*}
\label{eq:a}
P(u(t-t_1), \ldots,u(t-t_D)) &=& \sum_{\vec{r} \in \Delta_d^D} \alpha_{\vec{r}} u(t-t_1)^{r_1}\ldots u(t-t_D)^{r_D}
= \sum_{n=1}^{d} \sum_{\vec{r} \in \Delta_d^D(n)} \int_{\vec{\tau}\in \reals^n} h_n^{\vec{r}}(\vec{\tau}) \prod_{j=1}^{n}u(t-\tau_j)d\vec{\tau}
\\
&=& \sum_{n=1}^{d} \int \left [\sum_{\vec{r} \in \Delta_d^D(n)}  h_n^{\vec{r}}(\vec{\tau}) \right] \prod_{j=1}^{n}u(t-\tau_j) d\vec{\tau}
\\
&\defeq & \sum_{n=1}^{d} \int h_n(\vec{\tau})\prod_{j=1}^{n}u(t-\tau_j) d\vec{\tau}
\end{eqnarray*}

Therefore $H_0=0$ and the $n^{th}$-order GFRF is 
\begin{equation}
\label{eq:Hn once}
H_n(\Omgn) = \sum_{\vec{r} \in \Delta_d^D(n)} \Fourier{h_n^{\vec{r}}}(\Omgn) = \sum_{\vec{r}}  \alpha_{\vec{r}} \exp(-i\cdot  t_1 \sum_{j=1}^{r_1} \omega_j)\ldots \exp(-i \cdot t_D \sum_{j=n-r_D+1}^n \omega_j)
\end{equation}

The same approach is used with $\Nhist_{[a,b]}$.

$\bullet\underline{\Nmin (u,v)(t) = \min\{u(t),v(t)\}}$.
Here we must use a separable approximation of the form  $\Nmin(u,v) \approx \Uc u + \Vc v$.
This avoids product terms involving $u$ and $v$ which cause the loss of the nice GFRF representation of Volterra kernels~\cite{Billings03MISOVolterra}.
The Volterra operators $\Uc$ and $\Vc$ are obtained, again, by polynomial fitting.
Specifically, $\Uc u (t) = R(u(t))$ for a polynomial $R$ and $\Vc v (t) = Q(v(t))$ for a polynomial $Q$.
Both polynomials have a 0 constant term since zero inputs produce a zero output from $\Nmin$.
Note also that only the present value of the signal, $u(t)$, is used, since  it doesn't make sense to use past values $u(t-\tau)$ when approximating the instantaneous min operator.
The coefficients of the polynomials are obtained by least-squares as before.
Once the coefficients of $R$ and $Q$ are calculated, the following easily established proposition gives the kernels of the equivalent Volterra series.
\begin{proposition}
	\label{prop:poly volterra}
	The polynomial operator defined by $Nu(t) = \sum_{0\leq k\leq d} \alpha_{k} u(t)^k$ has an exact Volterra representation given by $h_0 = \alpha_{0}$, $h_n(\tau_1,\ldots,\tau_n) = \alpha_n \delta(\tau_1)\ldots \delta(\tau_n),~n\geq 1$.
	The corresponding GFRFs are $H_0=2\pi \alpha_0\delta(\omega)$, $H_n(\Omgn) = \alpha_n~\forall \Omgn$.
\end{proposition}

This concludes the derivation of Volterra series for the basic operators.
The following theorem allows us to calculate the GFRFs of entire formulas.
Given $\Omgn \in \reals^n$ and $\vec{m} \in \Delta_n^k(n)$, we can divide $\Omgn$ into $k$ sub-vectors, $\Omgn = (\Theta_1,\Theta_2,\ldots ,\Theta_k)$, s.t. sub-vector $\Theta_j$ has length $m_j$.
Define the mixing matrix $S^{(k,n)}$ of dimensions $k$-by-$n$ whose $j^{th}$ row is $(\vec{0}_{1\times (m_1+\ldots+m_{j-1})}, \vec{1}_{1\times m_j}, \vec{0}_{1\times (m_{j+1}+\ldots+m_k)})$, 
so $S^{(k,n)}\Omgn = (\sum_{j=1}^{m_1} \omega_j,\sum_{j=m_1+1}^{m_1+m_2} \omega_j,\ldots, \sum_{j=n-m_k+1}^n \omega_j)^T$.

\begin{theorem}[\cite{Carassale10Volterra}]
	\label{thm:carassale}
	Let $\Ac,\Bc$ be Volterra operators with GFRFs $\{A_n\}_{n=1}^{n_\Ac}$ and $\{B_n\}_{n=1}^{n_\Bc}$ respectively. Then the operator $\Hc \defeq \Bc \circ \Ac$ has Volterra GFRFs given by 
	\begin{eqnarray*}
	H_0 &=& \sum_{k=0}^{n_\Bc} B_k(0)A_0^k \\
	H_n(\Omgn) &=&\sum_{k=1}^{n_\Bc} \sum_{\vec{m}\in \Delta_n^k(n)}B_k(S^{(k,n)}\Omgn) \prod_{j=1}^k A_{m_j}(\Theta_j),\quad n\geq 1
	\end{eqnarray*}
\end{theorem}

Thus for instance, to get the GFRF of $\formula = \LTLonce_{[0,0.5]}g$ for some atom $g$, we derive the GFRF $\{B_k\}$ of $\Nonce_{[0,0.5]}$ and $\{A_k\}$ of $g$, then compute the GFRF of $\Nlo_{\formula}$ using~Thm.~\ref{thm:carassale}. 

\subsection{Why is Approximating $\formula_1 \LTLsince_I \formula_2$ Different?}
\label{sec:apx since}
For convenience, we write $\psi = \formula_1 \LTLsince_I \formula_2$.
The robustness $\rob_\psi$ is an operator on $\Cc(\reals,\reals)$, and we have shown that it is approximable by a Volterra series.
However it is constructed out of operators that change the dimensions of the signals, which adds difficulties to the actual computation of the approximation.

Specifically: fix an interval $[a,b] \subset \reals_+$, $\varepsilon >0$ and $f\in \FilterSp$; 
let $u$ denote an arbitrary signal in $\SigSpace$ and let $y \in \Cc(\reals^2,\reals)$, i.e. a continuous bounded function from $\reals^2$ to $\reals$.
We define three operators: $\Nmax_{2\rightarrow 1}: \Cc(\reals^2,\reals) \rightarrow \Cc(\reals,\reals)$,  
$\Nmin_{1\rightarrow 2}: \Cc(\reals,\reals) \rightarrow \Cc(\reals,\reals^2)$,
and $\Nmin_{2\rightarrow 2}: \Cc(\reals^2,\reals^2) \rightarrow \Cc(\reals^2,\reals)$.
They are:

\begin{multicols}{3}
	\noindent
	\begin{eqnarray*}
		\Nmax_{2\rightarrow 1}y(t) &=& \max_{t-b\leq t'\leq t-a} y(t',t)
	\end{eqnarray*}
	\begin{eqnarray*}
		\Nmin_{1\rightarrow 2} u(t',t) &=& \min_{t'<t''\leq t}u(t'')
	\end{eqnarray*}
	\begin{eqnarray*}		
		\Nmin_{2\rightarrow 2}(u,y)(t',t) &=& \min\{u(t'), y(t',t)\}
	\end{eqnarray*}
\end{multicols}
The following relation holds:
\begin{equation*}
\rob_{ \formula_1 \LTLsince_{[a,b]}\formula_2} =  \Nmax_{2\rightarrow 1} \circ \Nmin_{2\rightarrow 2}(\rob_{\formula_2}, \Nmin_{1\rightarrow 2} \rob_{\formula_1})
\end{equation*}

The approximation of $\rob_\psi$ by Volterra series therefore requires the approximation of the above basic operators, then composing them. 
Multi-dimensional Volterra series exist (i.e., Volterra operators over $\Cc(\Re^n,\reals)$), e.g., see~\cite{Mitra96QuadraticVolterra}.
However what we have above are operators that change the dimensions of the signals.
Sandberg~\cite{Sanberg97Myopic} provides a generalization of ~\cite{Boyd85FMVolterra} which allows the approximation of certain operators that map $\Cc(\reals^n,\reals^m)$ to $\Cc(\reals^m,\reals)$.
However this still falls short of our needs because of the presence of $\Nmin_{1\rightarrow 2}$.

A `quick-and-dirty' way to produce a Volterra series representation of a \textit{given} formula $\psi$ with Since - that is, with given atoms and structure - is to approximate its input-output relation on a representative set of signals by fitting Volterra kernels.
However this requires a new fit every time we change atoms or formula structure. 
It does not provide a generic approximation that can be composed with others, as we did in Section~\ref{sec:volterra apx of festl}.



\begin{figure*}[t!]
	\centering
	\begin{subfigure}[t]{\textwidth}
		\centering
		\includegraphics[scale=0.07]{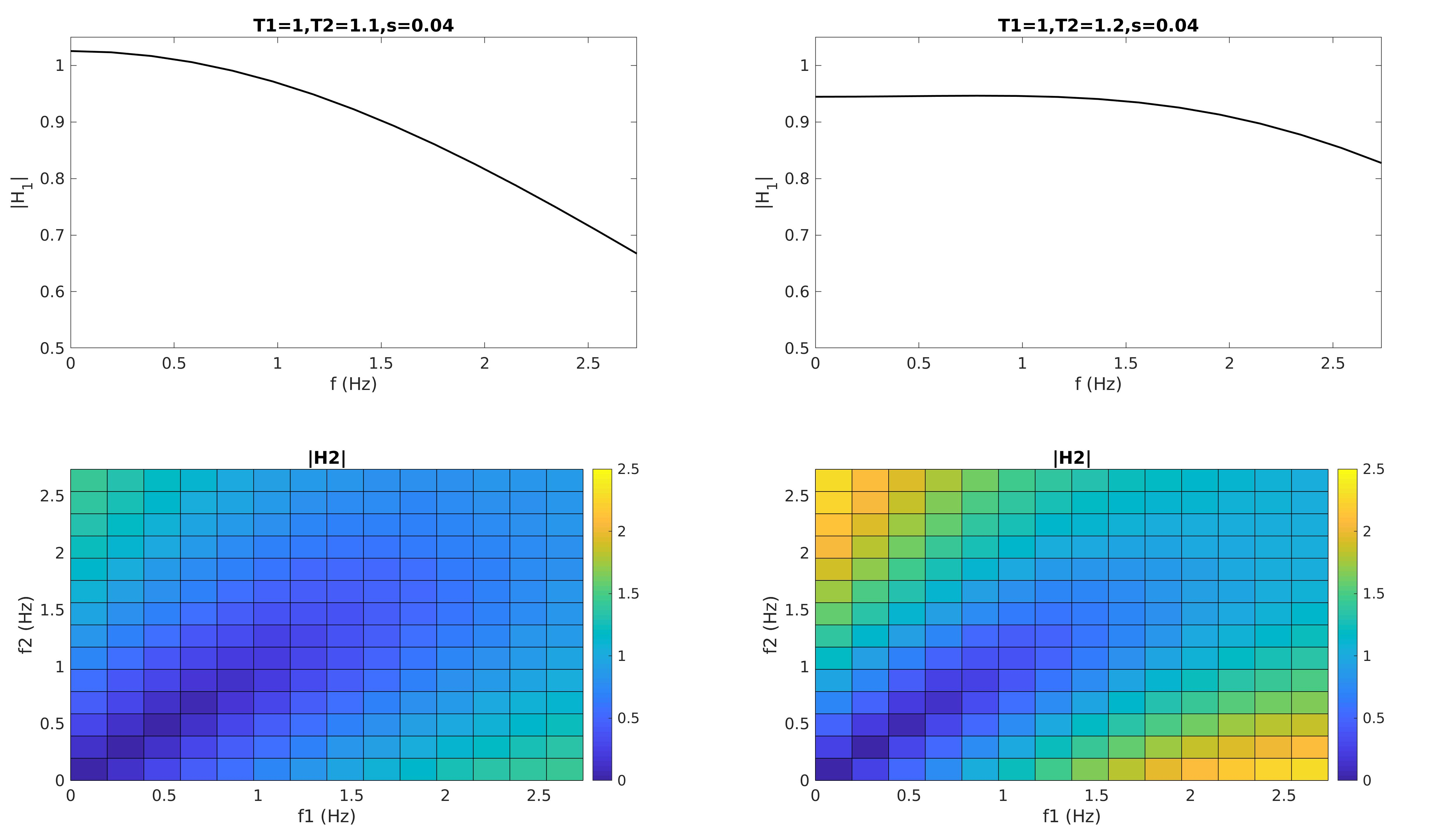}
		\caption{$|H_1|$ and $|H_2|$ for Historically$_{[1, T_2]}p_g$ with $T_2=1.1$ (left) and $T_2=1.2$ (right). The atom $g = G(0,0.04)$.}
		\label{fig:historically}
	\end{subfigure}%
	\\ \vspace{0.5cm}
	\begin{subfigure}[t]{\textwidth}
		\centering
		\includegraphics[width=0.8\linewidth]{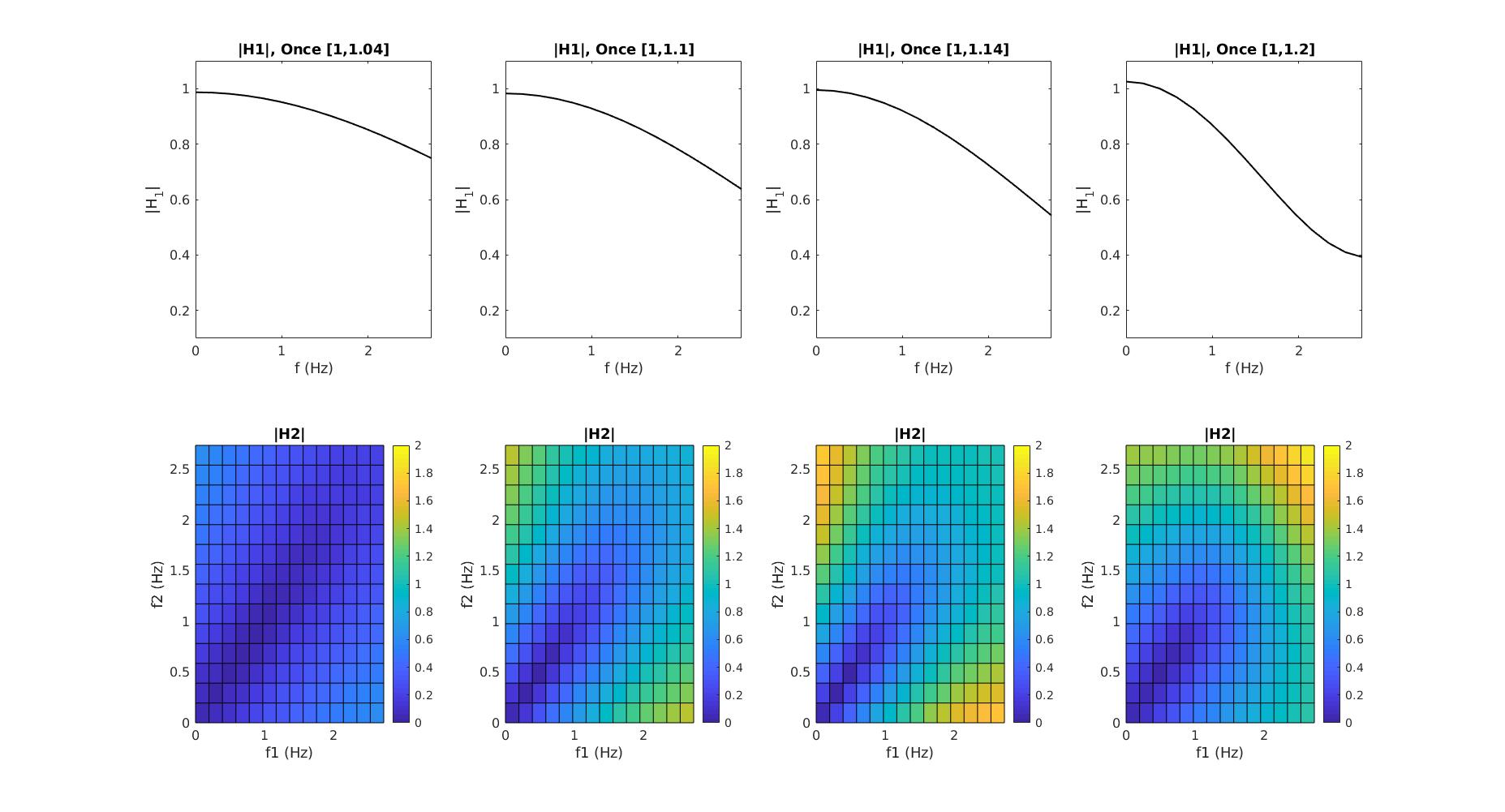}
		\caption{$|H_1|$ and $|H_2|$ for Once$_{[1, T]}p_g$ for four values of $T\in \{1.04,1.1,1.14,1.2\}$. The atom $g = G(0,0.04)$.}
		\label{fig:once T}
	\end{subfigure}
	\caption{GFRFs with varying temporal intervals. Color in digital copy.}
	\label{fig:varying T}
\end{figure*}

\begin{figure*}[t!]
	\centering
	\begin{subfigure}[t]{\textwidth}
		\centering
		\includegraphics[scale=0.07]{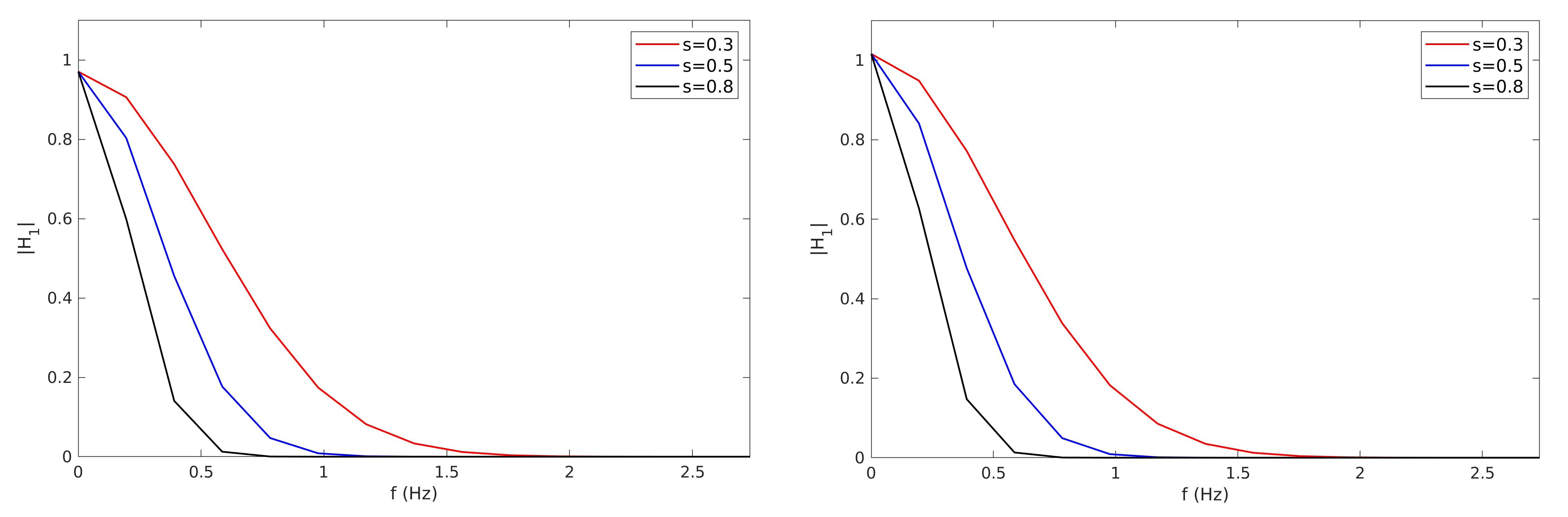}
		\caption{$|H_1|$ of Once$_{[1,1.1]} p_g$ (left) and Historically$_{[0,0.5]} p_g$ with $g=G(0,s), s=0.3,0.5,0.8$.}
		\label{fig:once hist H1 varying s}
	\end{subfigure}%
	\\ \vspace{0.5cm}
	\begin{subfigure}[t]{\textwidth}
		\centering
		\includegraphics[scale=0.065]{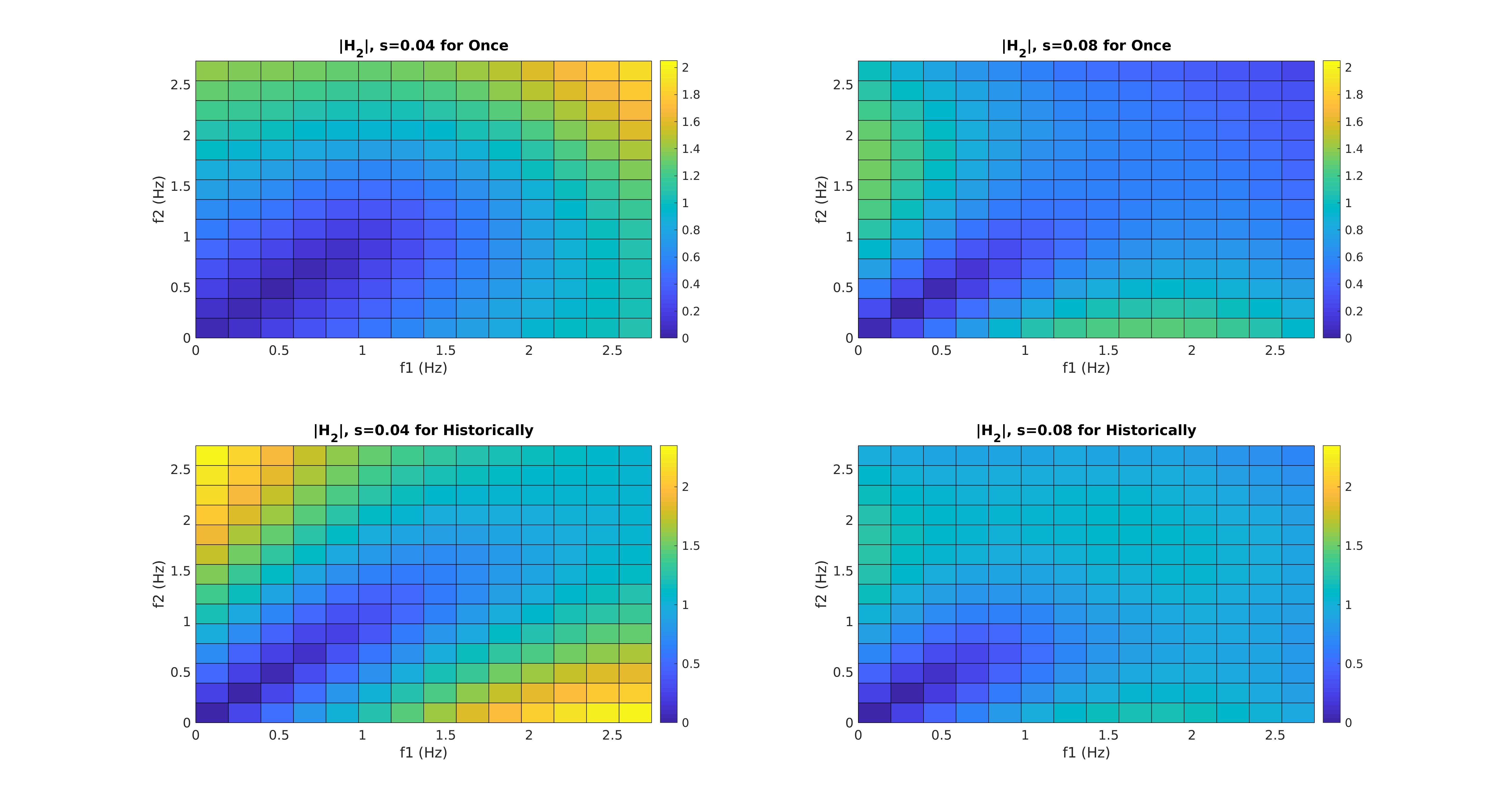}
		\caption{$|H_2|$ of Once$_{[1,1.1]} p_g$ (left) and Historically$_{[1,1.1]} p_g$ (right) with $g=G(0,s), s=0.04,0.08$.}
		\label{fig:once hist H2 varying s}
	\end{subfigure}
	\caption{Effect of support size $s$ for the atomic proposition filters. In (a), $s$ is larger than temporal interval width, which is 0.1. In (b) $s$ is much smaller. (Color in digital copy)}
	\label{fig:combined once hist varying s}
\end{figure*}
\section{Experiments: Fourier Analysis of Temporal Logic}
\label{sec:experiments}

\begin{figure*}[t!]
	\centering
		\includegraphics[scale=0.40]{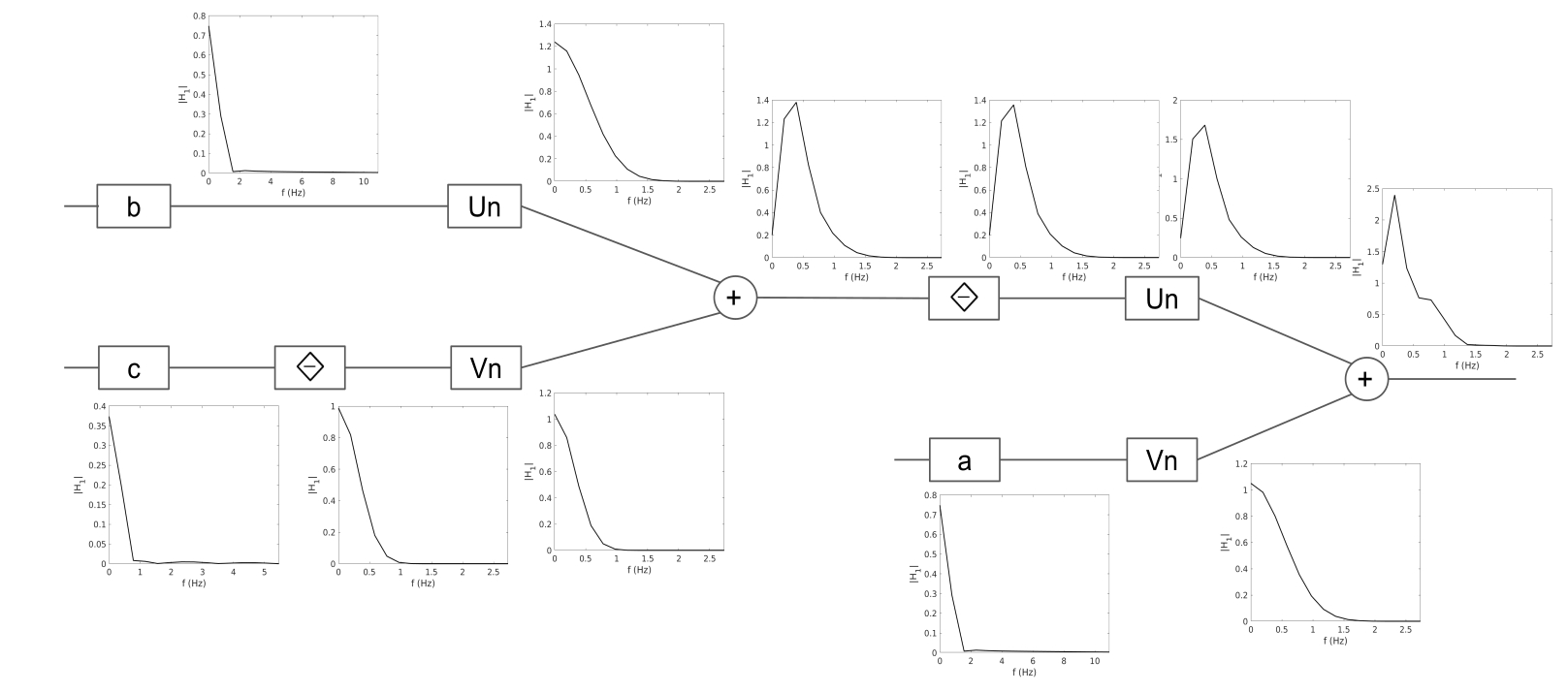}
		\caption{Block diagram of the Volterra representation of $\varphi_{T}$ given in~\ref{eq:phie}. Every displayed $|H_1|$ is the first-order spectrum of the entire composite formula up to that point. $Un$ and $Vn$ are the GFRFs of the separable Volterra operators $\Uc,\Vc$ that approximate $\Nmax$ (Section~\ref{sec:volterra formulas}.)}
		\label{fig:phie diagram}
		\vspace{0.5cm}
		\includegraphics[width=1.0\linewidth]{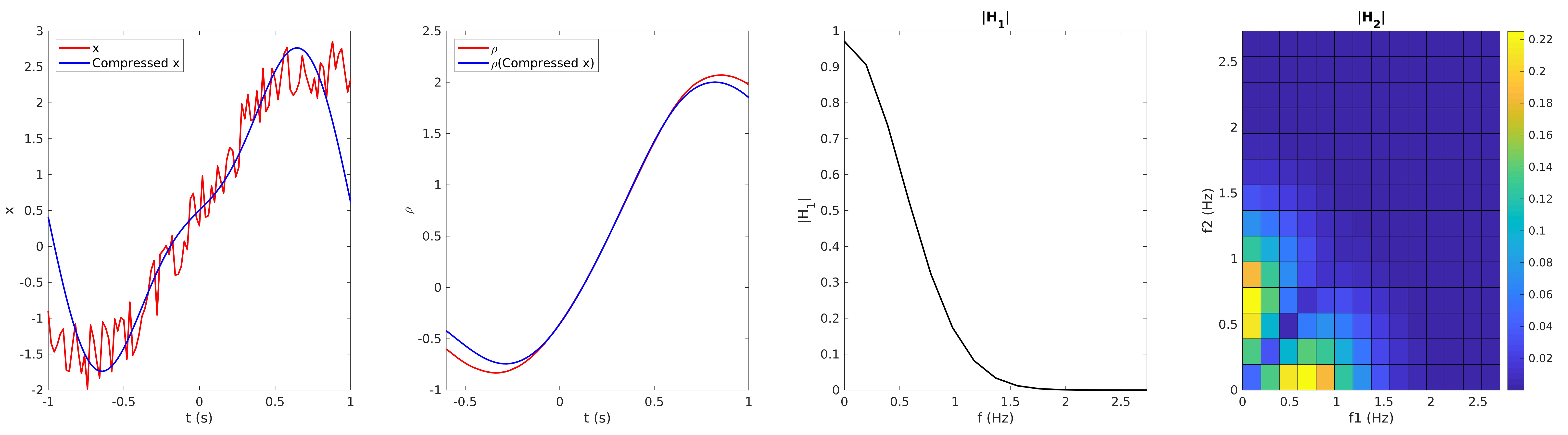}
		\caption{Filtering signals without accounting for downstream logic monitors leads to incorrect monitoring results. The frequency responses (right two panels) indicate a safe cut-off frequency around 1.5Hz. If an upstream low-pass filter applies a cut-off of 0.5Hz (left panel), the robustness signal is significantly changed (second panel). In particular, the truth values differ between red (original) and blue (post-filtering). (Colors in digital copy).}
		\label{fig:comp effect once bad}
\end{figure*}

We implemented the above calculations in a toolbox which we'll make available with the paper.
In this section we demonstrate the derivation of Generalized Frequency Response Functions for temporal logic robustness operators. 
In all experiments, the GFRFs were generated by solving appropriate versions of~\eqref{eq:lin sys} with degree-4 polynomials and test signals generated as random combinations of sinusoids.
Sinusoids are dense in $\Cc(\reals,\reals)$ so approximating the operators on sinusoids is a sensible thing to do.
The approximation error in all cases was in the order of $10^{-12}$.
That said, our objective here is not to provide the most efficient or the most general approximation scheme - that is for future work.

We reiterate that the Volterra approximations are \textit{not} meant to replace the monitoring algorithms that exist.
They are used as \textit{analysis tools} that provide a rigorous quantitative Fourier analysis of temporal logic: one that does not depend on intuition, is automatic, and such that once the GFRFs of a formula are obtained, the formula (and its monitor) are treated as just another signal processing box.

In what follows, $g = G(\mu, s)$ means that $g$ is a Gaussian measurement kernel with mean $\mu$ and standard deviation $s$.
 
\subsection{GFRFs of \festl~formulas}
\label{sec:comp rules}
 
\quad~~$\bullet$ We first consider the spectra of $\LTLhistorically_{[1,T_2]}p_g$ shown in Fig.~\ref{fig:historically}, with $g = G(0,0.04)$.
Increasing $T_2$ has a first-order effect (observed in $H_1$) of distributing the energy more uniformly over the range [0, 2.5]Hz, and suppressing less the higher frequencies.
$|H_2|$ on the other hand shows a more complex picture: while there's an increase of magnitude at higher values of $f_1$ or $f_2$ (top left and bottom right corners), the increase at higher $f_1$ \textit{and} $f_2$ is less marked.
 
 $\bullet$ Consider next the formula $ \LTLonce_{[0,T]}p$ for a fixed atom $p$, shown Fig.~\ref{fig:once T}.
 As $T$ increases, $H_1$ becomes more low-pass, but $H_2$ becomes more high-pass! This emphasizes the need to study all orders of the response, not only the linear first-order response.
 
$\bullet$ We now study the effect of using non-instantaneous measurements. 
Fig.~\ref{fig:once hist H1 varying s} shows the spectra $H_1$ of $\LTLonce_{[0,0.5]}p_g$ and $\LTLhist_{[0,0.5]}p_g$ where $g=G(0,s)$ for three values of $s$.
As $s$ increases, the Gaussian atom acts more like a low-pass filter (the measurement is lower resolution) and the overall formula has a more low-pass nature. 
By the same token, high-frequency noise is ignored by the formula and does not affect the monitoring verdict. 
Similarly, the $2^{nd}$-order spectra for these two formulas are shown in Fig.~\ref{fig:once hist H2 varying s} with increasing $s$.

In practice, the filter $f$ used in atomic propositions is imposed by the application and is derived from first-principles modeling of the physics of the system.
This Fourier analysis allows us to trace these effects quantitatively.

$\bullet$ Consider now the more complex formula $\formula_T$, which says that $a$ is true, preceded by $b$ $T$ units earlier, preceded by $c$ $T$ units earlier than that..
Here $a$, $b$ and $c$ are atoms with Gaussian filters of various widths.
\begin{equation}
\label{eq:phie}
\formula_{T} = a \land (\LTLonce_{[0,T]}(b \land \LTLonce_{[0,T]}c))
\end{equation}
It is not possible to read, from the formula, how the frequency responses of the various sub-monitors (for the sub-formulas) interact or cancel each other out. 
By contrast, Fig.~\ref{fig:phie diagram} shows the signal block diagram for computing this formula's Volterra series. 
This can be read as just another signal processing chain with non-linear filters.
On top of each box, we display the GFRF $H_1$ of the entire chain up to and including that box. 
This shows how the relevant frequencies evolve with the addition of each monitoring component).

\subsection{Compression's effect on monitoring}
\label{sec:compression}
We now illustrate what happens if attention is not paid to the frequency representation of temporal logic formulas when designing compression or filtering algorithms.
In Fig.~\ref{fig:comp effect once}, we The proposed method can be used as a signal processing tool for frequency domain analysis of temporal logic monitors but it is not supposed to replace the monitors themselves. Thus,it is more of an offline analysis tool that can be used to design filters which respect the  monitorability requirements. had shown how knowledge of the GFRFs allows us to perform monitoring-safe compression: even though the post-compression signal is markedly different from the original $x$, the monitoring results for the two signals were almost identical.

By contrast, in Fig.~\ref{fig:comp effect once bad}, we show the same signal but now compressed (via low-pass filtering) without regard to the GFRF or the monitored formula. 
The resulting monitoring result (in blue) is significantly affected, and the truth value (determined by checking where $\robf$ is positive or negative) is modified. 


\section{Conclusions}
\label{sec:conclusion}
We have presented a Fourier analysis of temporal logic using Volterra approximations of the robustness operators. 
Doing so has necessitated re-defining the semantics of atomic propositions using bounded-bandwidth filters, which led us to introduce the logic Bounded-Bandwidth STL. 
Using this analysis, it is possible to incorporate temporal logic monitors into signal processing chains as `just another' signal processing box.
 
Future work will seek to relax the constraints on the signal space.
In particular, we conjecture that it is possible to remove the compact-support requirement.
We will also seek more general approximations of the basic operators and extend them to Since.
Finally, the frequency representation in this paper presents a unifying formalism which we will leverage for optimal filter design that is monitoring-safe, i.e., that does not remove any signal content that is relevant to the output robustness signal.

\section*{Acknowledgments}
The authors would like to thank the anonymous reviewers for helpful comments regarding the definition of kernel space.

\bibliographystyle{splncs04}
\bibliography{iccps2017,hscc17,hscc2016,fainekos_bibrefs,hscc19,cav2019,colin_bib}
\end{document}